\documentclass[runningheads]{llncs}
\usepackage[T1]{fontenc}
\usepackage{graphicx}
\usepackage{amsmath}
\usepackage{amssymb}
\usepackage{listings}
\usepackage{xcolor}
\usepackage{url}
\usepackage{graphicx} 
\usepackage{xcolor}
\usepackage[scaled=0.8]{beramono}
\usepackage[T1]{fontenc}
\usepackage{array}
\usepackage{cleveref}
\usepackage{centernot}
\usepackage{mathtools}
\usepackage{stmaryrd}
\usepackage{mathpartir}
\usepackage{relsize}
\usepackage{adjustbox}
\usepackage{caption} 
\usepackage[numbers,sort&compress]{natbib}
\usepackage{enumerate}

\newcommand{\m}[1]{\mathsf{#1}}
\newcommand{\mb}[1]{\mathbf{#1}}
\newcommand{\mc}[1]{\mathcal{#1}}
\newcommand{\mt}[1]{\mathtt{#1}}
\newcommand{\mi}[1]{\mathit{#1}}

\newcommand{\ichoice}[1]{\oplus\{#1\}}
\newcommand{\tensor}{\otimes}
\newcommand{\one}{\mathbf{1}}
\newcommand{\texists}[2]{\;?\{#1\}.#2}
\newcommand{\tforall}[2]{\;!\{#1\}.#2}

\newcommand{\eclose}[1]{\m{close} \; #1}

\newcommand{\st}{<:}

\newcommand{\rt}[2]{\m{#1}[#2]}

\newcommand{\echoice}[1]{\&\{#1\}}

\newcolumntype{C}{>{$}c<{$}} 
\newcolumntype{L}{>{$}l<{$}} 
\newcolumntype{R}{>{$}R<{$}}

\newcommand{\semi}{\mathrel{;}}

\makeatletter
\newcommand{\xMapsto}[1]{\ext@arrow 0599{\Mapstofill@}{}{#1}}
\def\Mapstofill@{\arrowfill@{\Mapstochar\Relbar}\Relbar\Rightarrow}
\makeatother

\newcommand{\lolli}{\multimap}

\newcommand{\fwd}[2]{{#1} \leftrightarrow {#2}}
\newcommand{\fresh}[1]{(#1 \; \m{fresh})}

\newcommand{\ecase}[3]{\m{case} \; #1 \; (#2 \Rightarrow #3)}
\newcommand{\esend}[2]{\m{send} \; #1 \; #2}
\newcommand{\erecv}[2]{#2 \leftarrow \m{recv} \; #1}
\newcommand{\ewait}[1]{\m{wait} \; #1}
\newcommand{\espawn}[3]{#1 \leftarrow #2 \; #3}
\newcommand{\eassert}[2]{\m{assert} \; #1 \; \{#2\}}
\newcommand{\eassume}[2]{\m{assume} \; #1 \; \{#2\}}

\newcommand{\with}{\,\&\,}

\definecolor{dkgreen}{rgb}{0,0.4,0}
\definecolor{ltblue}{rgb}{0,0.4,0.4}
\definecolor{dkblue}{rgb}{0,0,0.6}
\definecolor{dkviolet}{rgb}{0.3,0,0.5}
\definecolor{codegreen}{rgb}{0,0.6,0}
\definecolor{eminence}{RGB}{108,48,130}

\lstdefinelanguage{prest}{
  morekeywords=[1]{
    define, proc, yield, type, decl
  },
  morekeywords=[2]{
    flip, match, case, pcase, send, recv, assert, assume, wait, close
  },
  morekeywords=[3]{},
  tabsize=2,
  sensitive=true,
  breaklines=false,
  extendedchars=true,
  alsoletter=*,
  captionpos=b,
  morecomment=[l]{//},
  mathescape=true,
  basicstyle=\linespread{1}\ttfamily,
  identifierstyle={\ttfamily\color{black}},
  keywordstyle=[1]{\ttfamily\color{dkblue}},
  keywordstyle=[2]{\ttfamily\color{dkviolet}},
  keywordstyle=[3]{\ttfamily\color{ltblue}},
  stringstyle=\ttfamily,
  commentstyle={\ttfamily\color{dkgreen}},
  literate=
  {|-}{$\vdash\ $}{2}
  {<-}{$\ \leftarrow\ $}{2}
  {=>}{$\ \Rightarrow\ $}{3}
  {<->}{$\ \leftrightarrow\ $}{2}
  {|\{*\}-}{$\mathlarger{\mathlarger{\vdash^*}}$}{3}
}

\newcommand{\lst}[1]{\lstinline[mathescape,language=prest]!#1!}

\newcommand{\num}[1]{\overline{#1}}
\newcommand{\eif}[3]{\m{if} \; #1 \; \m{then} \; #2 \; \m{else} \; #3}

\begin{document}
\title{Practical Refinement Session Type Inference (Extended Version)}
%
%
%
\authorrunning{T. Ueno \and A. Das}
\author{
    Toby Ueno\inst{1}
    \and 
    Ankush Das\inst{2}
}

\institute{
    University of Edinburgh, UK\\
    \email{toby.ueno@ed.ac.uk}
    \and
    Boston University, Boston, MA, USA \\
\email{ankushd@bu.edu}}
%
\maketitle              
\begin{abstract}
Session types express and enforce safe communication in concurrent message-passing systems by
statically capturing the interaction protocols between processes in the type.
Recent works extend session types with arithmetic refinements, which enable additional
fine-grained description of communication, but impose additional annotation burden on the programmer.
To alleviate this burden, we propose a type inference algorithm for a session type system
with arithmetic refinements.
We develop a theory of subtyping for session types, including an algorithm which we prove sound
with respect to a semantic definition based on type simulation.
We also provide a formal inference algorithm that generates type and arithmetic constraints,
which are then solved using the Z3 SMT solver.
The algorithm has been implemented on top of the Rast language, and includes 3 key optimizations
that make inference feasible and practical.
We evaluate the efficacy of our inference engine by evaluating it on 6 challenging benchmarks,
ranging from unary and binary natural numbers to linear $\lambda$-calculus.
We show the performance benefits provided by our optimizations in coercing Z3 into solving the
arithmetic constraints in reasonable time.
\end{abstract}

\section{Introduction}

Binary session types~\cite{hondaDyadicInteraction1993,Honda98esop,Honda08POPL,cairesSessionTypes2010,Caires16mscs}
provide a structured way to express communication protocols
in concurrent message-passing systems.
Types are assigned to bi-directional channels connecting processes
and describe the type and direction of communication.
Type checking then ensures that processes on either end of the channel
respect the interaction described by the channel type, i.e., they follow
both the type and direction.
As a result, standard type safety theorems guarantee session fidelity
(preservation), i.e., protocols are not violated at runtime and
deadlock-freedom (progress), i.e., no process ever gets stuck.
However, even in the presence of recursion, vanilla session types
can only specify very basic protocols which has led to
a number of extensions, such as context-free session types~\cite{Thiemann16ICFP,Silva23CONCUR},
label-dependent session types~\cite{Thiemann20POPL}, and general dependent session
types~\cite{WuX17Arxiv,Toninho11PPDP,Toninho21PPDP,Griffith13NFM}.

One such extension is to integrate session types with refinements~\cite{Das20FSCD,dasRefinements2020,dasVerified2020},
in the style of DML~\cite{Xi99popl,pfenningRefinementTypes} to enable
lightweight verification of message-passing programs.
For instance, consider the simple (unrefined) session type for naturals:
\[
    \mi{type} \; \m{nat} \triangleq \ichoice{\mb{succ} : \m{nat}, \mb{zero} : \one }
\]
The communication protocol governed by this type is that the sender can either
produce the $\mb{succ}$ label or the $\mb{zero}$ label, on which the receiver
must branch upon.
However, the type does not give any information beyond this basic behavior.
In particular, the type does not define exactly how many $\mb{succ}$ labels
will be produced by the type.
To address this limitation, one can refine the type:
\[
    \mi{type} \; \m{nat}[n] \triangleq \ichoice{\mb{succ} : \; ?\{n > 0\}.\, \m{nat}[n-1], \mb{zero} :\; ?\{n = 0\}.\, \one }
\]
The type is indexed by a natural number $n$ that denotes the number of $\mb{succ}$
labels produced by the sender.
On sending the $\mb{succ}$ label, the type $?\{n > 0\}.\, \m{nat}[n-1]$ describes
that the sender must produce a \emph{proof} that $n > 0$, after which the type
transitions to $\m{nat}[n-1]$, meaning that the sender will now produce
$(n-1)$ $\mb{succ}$ labels.
Such refinements enable lightweight verification of protocol processes;
for instance, one can define that an $\m{add}$ process can take two naturals
$\m{nat}[m]$ and $\m{nat}[n]$ to produce $\m{nat}[m+n]$.

As useful as refinements are, they introduce a significant syntactic burden
on the programmer.
First, programmers need to annotate process declarations (like $\m{add}$)
with the refined type of all channels connected to it.
Second, they need to add explicit assertions and assumptions in their programs
for successful type checking.
Finally, there is usually very little feedback from the compiler in the
event that the type annotations are incorrect.
As programs get more complicated, so do the annotations, further exacerbating
the cost of this overhead.

To address this syntactic overhead, this paper proposes an \emph{automatic
type inference} algorithm for refinement session types.
The goal is to infer the (possibly recursive) type of processes from their definitions.
However, this task poses several challenges, both in theory and implementation.
First, type checking, and therefore, type inference is \emph{undecidable}~\cite{dasRefinements2020}
even in the presence of simple linear arithmetic refinements.
Thus, any practical algorithm must be approximate and rely on appropriate
heuristics.
Second, session types follow \emph{structural} and \emph{equi-recursive} typing;
type equality does not depend on type names, rather the structure of types.
Moreover, there are no explicit folds or unfolds in programs to handle recursion.
This more expressive nature of typing introduces complications for the inference engine,
which needs to keep track of the type structure based on its communication. 
Third, the most \emph{general} refined type for a process is often not
the most useful one.
As an example, consider a process that produces $\m{nat}[0]$ by
sending a $\mb{zero}$ label and terminating.
The most general type of this process is $\ichoice{\mb{zero} : \one}$,
which is a \emph{subtype} of $\m{nat}[0]$, but probably not the most useful
type.
On the other hand, producing \emph{any} supertype of the most general type
is also not satisfactory.
As an example, we can infer many possibilities for index arguments;
for example, a correct type for the $\m{add}$ process is to take
$\m{nat}[2m]$ and $\m{nat}[n]$ to produce $\m{nat}[2m+n]$.
Although correct, this is not the most precise type for $\m{add}$.
In particular, this type would not work for numbers where the
first argument is odd.
Therefore, the most general type is not always useful— the type that is reported to the user must be general enough
to accommodate a large set of programs.

To infer a more realistic type for processes, we
extend the notion of subtyping in session types
to arithmetic refinements.
To that end, we adopt and generalize the declarative definition for subtyping
introduced in seminal work by~\citet{gaySubtypingSessionTypes2005}.
Using this definition as our foundation, we introduce a set of algorithmic
subtyping rules that produce a set of type constraints from which
the most general type can be inferred.
We show that this algorithm is \emph{sound}: if $A$ is a subtype of $B$
according to our algorithm, it must be so under our definition.
Next, we introduce a range of heuristics in our inference engine that carefully balance
precision and practicality to achieve user-friendly types.

We have implemented our type inference engine on top of the Rast~\cite{dasWorkAnalysis2018,Das18ICFP,Das20FSCD}
programming language.
Rast is a language targeted towards analyzing parallel and sequential complexity
of session-typed message-passing programs.
Rast supports type refinements and type checking, but not type inference.
Our implementation performs inference in two stages: first inferring the base
session type for each channel, and then inferring the type indices.
Separating out these two stages not only simplifies the inference engine
but rejects ill-typed programs early.
The first stage of inference collects the subtyping constraints and solves them
using a standard unification algorithm.
The second stage extracts the arithmetic constraints from subtyping and ships them
to the Z3 SMT solver~\cite{Moura08Z3}.
To improve the performance of type inference, we introduce three key
\emph{optimizations} in both stages: \emph{(i) transitivity} to eliminate
the intermediate types in a process definition,
\emph{(ii) polynomial templates} to reduce the search space of arithmetic
expressions for Z3, and \emph{(iii) theory of reals} to find a satisfying
assignment faster, which is then converted to natural numbers, if possible.

We analyze the efficacy of our inference engine and the performance improvements
provided by the optimizations using six benchmarks from the Rast language.
These include unary and binary natural numbers indexed by their value, lists
indexed by their size, and linear $\lambda$-calculus expressions indexed by
the size of the term.
We infer the type of a few standard processes, e.g., standard arithmetic operations
on numbers like addition and doubling; standard list operations like append and split;
and evaluation of expressions in the calculus.
Our experiments reveal that transitivity provides an order of magnitude performance
benefit in the first stage.
Moreover, both polynomial templates and theory of reals are important to make
the second stage feasible.
Without these two optimizations, Z3 times out on even the simplest of examples.
We conclude that all three optimizations are necessary to make inference scalable
and practical.

To summarize, the paper makes the following contributions:
\begin{itemize}
    \item A declarative definition of subtyping and a sound algorithm for subtyping of
    refinement session types (Section~\ref{sec:subtyping}),
    \item A type inference algorithm that generates typing and arithmetic constraints
    and its proof of soundness (Section~\ref{sec:inference}),
    \item An implementation (Section~\ref{sec:impl}) of our inference engine along with
    the three optimizations, and
    \item An evaluation (Section~\ref{sec:eval}) of our algorithm on six challenging benchmarks.
\end{itemize}
This extended version of the paper 
includes an appendix, which
contains the complete set of subtyping and inference rules
along with proofs of soundness.

\vspace{-1em}
\section{Motivating Examples}
\label{sec:overview}
To highlight the main challenges underlying our type inference engine, we introduce a
series of motivating examples following the syntax of the Rast language~\cite{Das20FSCD,dasVerified2020}.
Our setting is built on \emph{binary} session types~\cite{cairesSessionTypes2010,Caires16mscs}
which consists of \emph{processes} that exchange messages along bidirectional \emph{channels}.
Processes are classified into \emph{providers} and \emph{clients} making up the two endpoints of a channel,
the former providing the channel to the latter who consumes it.
Generally speaking, the client creates a new channel by \emph{spawning} the corresponding provider
which are then connected by this freshly created channel.
This also establishes a \emph{parent-child} relationship between the processes; the client
acts as the parent while the newly spawned provider acts as the child.
In this sense, a client process can have many children, i.e., many channels connecting it to providers
but a provider can have a single client, i.e., parent.
This parent-client connection generalizes to a tree-shaped topology as the number of processes increase,
but never a cycle, thereby preventing deadlocks.

This duality is expressed by generalizing the linear sequent $A_1, A_2, \ldots, A_n \vdash C$
and assigning channel names to each proposition as follows:
\[
  (x_1 : A_1), (x_2 : A_2), \ldots, (x_n : A_n) \vdash (z : C)
\]
The above judgment describes a process that is a client of channels $x_1, \ldots, x_n$
of corresponding types $A_1, \ldots, A_n$ and a provider of channel $z$ of type $C$.
To communicate, the provider and client perform dual matching actions on the channel,
and this action is exactly governed via \emph{session types}.
Just as a channel's meaning changes depending on whether it is being provided or consumed,
i.e., left or right of the sequent,
so does the type of the channel.
As processes communicate on channels, the type of the channel evolves, thereby capturing
not a single communication but a \emph{session} of messages.

Revisiting the example of natural numbers, the basic type is defined as
\[
	\mi{type} \; \m{nat} = \ichoice{\mb{zero}: \one,\ \mb{succ}: \m{nat}}
\]
The $\oplus$ type is associated with a set of labels $\{\mb{zero}, \mb{nat}\}$ of which
the channel provider can choose one to send.
The type after the message transmission is indicated by the colon: if the provider
sends $\mb{zero}$, the continuation type is $\one$, indicating termination and closing
of the channel;
if instead the provider sends $\mb{succ}$, the type recurses back to $\m{nat}$, meaning
the protocol simply repeats.

Processes can be defined in Rast using a \emph{declaration} describing the process type
and a \emph{definition} describing the implementation.
We follow the declaration and definition of a process called \lst{two}
that represents the number two:
\begin{lstlisting}[language=prest]
decl two : . |- (x: nat)
proc x <- two = x.succ; x.succ; x.zero; close x
\end{lstlisting}
The first line shows the declaration: the left of the turnstile ($\vdash$) shows the
channels (and their types) used by the process, while the right shows the offered
channel and type. The \lst{two} process does not use any channels, hence a dot ($\cdot$)
on the left and offers channel \lst{x} (denoted by \lst{<-}) of type \lst{nat}.
The process definition for \lst{two} (beginning with $\mt{proc}$) expresses that
the process sends exactly two $\mb{succ}$ labels on channel $x$ (\lst{x.k} denotes sending
label \lst{k} on channel \lst{x}) followed by a $\mb{zero}$ label, and then ultimately
terminates by closing the channel \lst{x}.
Without type inference, Rast only supports type checking,
meaning the programmer is required
to provide both the process declaration and definition.
The goal of this work is to automatically generate the declaration given the definition,
in the style of languages like OCaml and SML.

\paragraph*{\textbf{Structural Subtyping in Inference}}
At first, we might try a na\"{i}ve approach at type inference to produce the most general type.
For instance, if we are sending label \lst{k} on offered channel \lst{x}, the most
general type of \lst{x} must be $\ichoice{k: A_k}$, where $A_k$ is the type of channel \lst{x}
after the communication.
If we apply this technique on the process \lst{two} above, we will arrive at the following type:
\begin{lstlisting}[language=prest]
proc x <- two = x.succ;         % x: $\ichoice{\mb{succ} : \ichoice{\mb{succ} : \ichoice{\mb{zero} : \one}}}$
	              x.succ;         % x: $\ichoice{\mb{succ} : \ichoice{\mb{zero} : \one}}$
                x.zero;         % x: $\ichoice{\mb{zero} : \one}$
                close x         % x: $\one$
\end{lstlisting}
It is easier to follow this type inference bottom-up to build the inferred type incrementally.
The last operation on \lst{x} is closing, hence its type at that point must be $\one$,
as indicated by the comment on the right.
Next, we send the label $\mb{zero}$, meaning the most general type for \lst{x} must be
$\ichoice{\mb{zero} : \one}$ because this is the simplest type that allows sending
of label $\mb{zero}$.
Following this intuition for the next two labels, we get the type of \lst{x} to be
$\ichoice{\mb{succ}: \ichoice{\mb{succ}: \ichoice{\mb{zero}: \one}}}$.

This type however is not satisfactory for a variety of reasons.
First, this type is not generalizable: the process types for \lst{one}, \lst{two}, \lst{three}, \ldots
will all be different which, in turn, hampers code reuse.
Also, a programmer would be more inclined to use $\m{nat}$ as the process type.
Second, this type would not be usable for a larger process like $\m{add}$ which
needs to be defined for generic natural numbers, and not every particular number and type.
Second, although the generated type (or the \emph{most general type}) is correct for typechecking
purposes, using this type in type error messages would be unncessarily verbose and will progressively
become more incomprehensible as programs grow in complexity.

Evident from this is the necessity of a notion of \textit{subtyping}:
we would like to have a system where $\ichoice{\mb{succ}: \ichoice{\mb{succ}: \ichoice{\mb{zero}: \one}}}$
is a subtype of $\m{nat}$.
This subtyping system must also account for the \emph{structural} nature of types rather than
a \emph{nominal} one.
This is in contrast with refinement type systems for functional languages like
Liquid Types~\cite{Rondon08PLDI} and DML~\cite{Xi99popl} where the base types can be inferred
using a standard Hindley-Milner inference algorithm.
As will be evident from our subtyping algorithm, there is no notion of base types in structural
type systems.
Concretely, this decision entails that message labels like $\mb{zero}$ and $\mb{succ}$ are not tied to the $\m{nat}$ type specifically, unlike e.g. the constructors of an algebraic datatype in most functional languages.
For instance, we can use the same constructors for multiple types as follows:
\[
  \mi{type} \; \m{even} = \ichoice{\mb{zero}: \one,\ \mb{succ}: \m{odd}} \qquad
  \mi{type} \; \m{odd} = \ichoice{\mb{succ}: \m{even}}
\]
And both $\m{even}$ and $\m{nat}$ are valid types for process \lst{two}.
Thus, our algorithm for subtyping must analyze the (possibly mutually
recursive) structures of types.
Our approach is to infer the most general type that is \emph{provided in the signature} by the programmer.
With subtyping in place, we can say that a process
offers a channel of the most general type \emph{or any valid supertype}, thereby significantly
increasing readability while maintaining correctness.

\vspace{-0.4em}
\subsection{Introducing Refinements}

The notion of subtyping becomes more involved in the presence of \emph{arithmetic refinements}.
For instance, revisiting the refinement of natural number type:
\[
  \m{type}\ \m{nat}[n] = \ichoice{\mb{zero}: \texists{n=0}{\one},\ \mb{succ}: \texists{n>0}{\m{nat}[n-1]}}
\]
The refinements enforce that the provider is only permitted to send the $\mb{zero}$ label when $n=0$,
and the $\mb{succ}$ when $n>0$, respectively.
By specifying refinements on types, we are able to convey more type-level information about our processes:
for instance, we can guarantee that the amended $\texttt{two}$ process will send \emph{exactly}
two $\mb{succ}$ messages if the offered type is $\m{nat}[2]$.
\begin{lstlisting}[language=prest]
decl two : . |- (x: nat[2])
proc x <- two = x.succ; assert x {2 > 0}; 
                x.succ; assert x {1 > 0};
                x.zero; assert x {0 = 0}; close x
\end{lstlisting}
We describe the implementation of the process to show how the proof obligations are upheld.
Each time we send a $\mb{succ}$ message, we need to send a proof that $n > 0$, which is achieved
using an assertion.
Therefore, the process first asserts that $\{2 > 0\}$ after sending $\mb{succ}$
when the type of \lst{x} is $\m{nat}[2]$.
After this, the type transitions to $\m{nat}[1]$, hence the process sends another $\mb{succ}$
followed by asserting that $\{1 > 0\}$.
Finally, the type transitions to $\m{nat}[0]$, hence the process sends the $\mb{zero}$ label
and an assertion that $\{0 = 0\}$ and closes.

As can be observed from this example, these assertions add a significant syntactic overhead
on the programmer.
To eliminate this burden, our inference algorithm also performs \emph{program reconstruction}
which inserts these assertions automatically.
Our inference engine deduces these assertions by following the type structure
of the refined nat type.
Remarkably, this reconstruction is performed effectively even if multiple types
(e.g., even, odd, nat) are using the same label constructors.

\paragraph*{\textbf{Refinement Inference.}}
In addition to inferring type names like $\m{nat}$, we must also infer the \emph{refinements} on those types. These refinements are not guaranteed to be numbers, as in the \lst{two} process; rather, they are arbitrary expressions which can involve free arithmetic variables. The following process, which adds two numbers, is shown after reconstruction, but before refinement inference:
\begin{lstlisting}[language=prest]
decl add[m][n]: (x: nat[e0(m,n)]) (y: nat[e1(m,n)]) |- (z: nat[e2(m,n)])
proc z <- add[m][n] x y = case x (
    zero => assume x {e0(m,n) = 0}; wait x; z <-> y
  | succ => assume x {e0(m,n) > 0};
        z.succ; assert z {e2(m,n) > 0};
        z <- add[m-1][n] x y )
\end{lstlisting}
The \lst{add} process case analyzes on \lst{x}. If \lst{x} sends the $\mb{zero}$ label
(meaning $e_0(m, n) = 0$), we simply \emph{identify} channels \lst{y} and \lst{z}
(equivalent to returning \lst{y}) in a functional setting.
On the other hand, in the $\mb{succ}$ branch, where we 
\emph{assume} that the refinement on \lst{x} is greater than 0,
we send the $\mb{succ}$ label on \lst{z} followed by asserting that $e_2(m, n) > 0$
(as required by $\m{nat}$ type).
Finally, we recurse by calling the $\m{add}$ process again.

Suppose we have already determined that each channel has type $\m{nat}$;
we must still find expressions $e_0$, $e_1$, and $e_2$ in variables $m, n$
such that all of our assertions hold.
We delegate this non-trivial task to Z3: we first identify any \emph{constraints}
on any candidate expressions, e.g. that $e_0(m, n)=0$ implies $e_2(m, n) = e_1(m, n)$ since channel $y$
is forwarded on to $z$ in the $\mb{zero}$ branch.
Similarly, in the $\mb{succ}$ branch, we get that $e_0(m, n) > 0$ implies $e_2(m, n) > 0$.
With these constraints, we query Z3 to find a satisfying assignment for our refinements.
In the case of \lst{add}, we would expect that $e_0(m, n)=m$, $e_1(m, n) = n$, and $e_2(m, n)=m+n$.

Naturally, our first attempt was to treat $e_i$'s as \emph{uninterpreted functions}, which Z3 has the ability
to solve for.
We simply shipped our constraints over to Z3, asking for a satisfying assignment.
In practice, however, this approach is not feasible: either the solver times out or it returns a
non-polynomial expression that cannot be expressed in our refinement layer, such as if-then-else constructs.
Thus, we first substitute each expression with a \emph{polynomial template}:
for instance, we transform $e_0(m,n)$ into $c_0m+c_1n+c_2$, and only ask Z3 to solve for the
coefficients $c_i$'s.
Any expression in our language takes this form, so we lose no generality, and we find that this
optimization greatly improves both the reliability and performance of Z3.
In fact, polynomial templating is only one such optimization which enables our algorithm to
succeed on reasonably complex examples; we detail others in Section \ref{sec:impl}.

\vspace{-1em}
\section{Background on Refinement Session Types}\label{sec:background}

We describe the language of session types upon which we implement type inference.
We first introduce the basic session types, which constitute the core of our communication protocols,
followed by the refinement layer, which extends the core with constructs to send and receive proofs
and witnesses. We include detailed examples of each type constructor in Appendix~\ref{appendix-ex}.

\vspace{-0.4em}
\subsection{Basic Session Types}

Our basic session type system is in correspondence with intuitionistic linear logic~\cite{cairesSessionTypes2010}.
The types allow exchange of labels, other channels, and close, i.e., termination messages.
The type and process syntax is defined as:
\begin{align*}
A, B ::=\ &\ichoice{\ell: A_\ell}_{\ell \in L} \mid \echoice{\ell: A_\ell}_{\ell \in L}
&& \text{(internal and external choice)}\\
\mid\ &A \tensor B \mid A \lolli B &&\text{(tensor and lolli)}\\
\mid\ &\one \mid V &&\text{(unit and type variable)}\\
P ::=\ &x.k \semi P  \mid \ecase{x}{\ell}{P_\ell}
&& \text{(send and receive labels)}\\
\mid\ &\esend{x}{e} \semi P \mid \erecv{x}{y} \semi P
&& \text{(send and receive channels)}\\
\mid\ &\eclose{x} \mid \ewait{x} \semi P
&& \text{(close and wait for close)}\\
\mid\ &\fwd{x}{y} \mid \espawn{x}{f}{\num{y}} \semi P
&& \text{(forward and spawn)}
\end{align*}

The structural types are divided into pairs of types— namely, $\oplus$ and $\with$ form a pair, as do $\tensor$ and $\lolli$.
The partner of a type exhibits the same sort of communication in the opposite direction.
For instance, a provider of a channel $x$ of type $\ichoice{\ell: A_\ell}_{\ell \in L}$,
must \emph{send} one of the labels $k \in L$ via the expression $x.k$ and continue as $A_k$.
Dually, if a provider offers $x : \echoice{\ell: A_\ell}_{\ell \in L}$, it \emph{receives}
a label in $L$, which it case-analyzes via the expression $\ecase{x}{\ell}{P_\ell}$;
the process continues as $P_\ell$ and the channel as $A_\ell$.
The other types, $\tensor$ and $\lolli$, are analogous, but for sending \emph{channels} instead of labels:
provider of $x : A \tensor B$ will send a channel $e$ of type $A$ via $\esend{x}{e}$.
Conversely, a provider of $x : A \lolli B$ will receive a channel of type $A$ which it binds to $y$
via $\erecv{x}{y}$.
Finally, a provider of $(x: \one)$ must use $\eclose{x}$ to terminate channel $x$ and the process,
while the client takes the form $\ewait{x} \semi P$, i.e., waits for $x$ to close and then continue as $P$.
The remaining pieces of syntax, $\fwd{x}{y}$ and $\espawn{x}{f}{\num{y}}$, refer to forwarding and spawning, respectively.
Forwarding $\fwd{x}{y}$ means that we identify channels $x$ and $y$ passing all messages between them,
and spawning a new process $f$ is analogous to a function call in other languages:
we provide $f$ with channels $\num{y}$ and receive a channel $x$ back.
Finally, note that a process \emph{offers} one channel and \emph{consumes} any number of other channels.
And the role of a type flips based on whether it is provided or consumed.

\vspace{-0.4em}
\subsection{Refinement Layer}

Types are refined using the following type and process constructs:
\begin{align*}
A ::=\ &\cdots\ \mid\ V[\num{e}] &&\text{(indexed type variable)}\\
\mid\ &\texists{\phi}{A} \mid \tforall{\phi}{A}
&&\text{(assertion and assumption)}\\
\mid\ &\exists n.A \mid \forall n.A 
&&\text{(quantifiers)}\\
P ::=\ &\cdots\ \mid\ \eassert{x}{\phi} \semi P \mid \eassume{x}{\phi} \semi P
&& \text{(assert and assume)}\\
\mid\ &\esend{x}{\{e\}} \semi P \mid \erecv{x}{\{n\}} \semi P_n
&& \text{(send and receive witnesses)}
\end{align*}

The crucial types for verifying program behavior are ?, and its dual, !.
When a process provides $(x: \texists{\phi}{A})$, it sends a \emph{proof} of $\phi$ along $x$ via $\eassert{x}{\phi}$,
and the client of $x$ receives this proof of $\phi$ via $\eassume{x}{\phi}$.
Note that no actual proof objects are sent at runtime; instead, we merely communicate that such a proof exists.
As before, the dual $(x: \tforall{\phi}{A})$ reverses the direction of communication,
allowing for the provider of $x$ to receive a proof and for a client to send one.
We also allow channels to send and receive \emph{witnesses} through the quantifier types $\exists$ and $\forall$.
A provider of $(x: \exists n.\, A)$ sends a witness expression $e$ via $\esend{x}{\{e\}}$, a client of
$x$ receives the value of $e$ for $n$ and substitutes $[e/n]A$ in the continuation
via $\erecv{x}{\{n\}}$.
Type $\forall n.\, A$ exhibits dual behavior.
We use these witnesses to communicate natural numbers which can themselves be used in future assertions
and assumptions. For instance, we could write $\exists k. \texists{n=2*k}{\m{nat}[k]}$ to signify that
$n$ is even and $k$ is its witness.

The language for arithmetic expressions $e$ and propositions $\phi$ is standard and described below.
We use $i$ to denote constant numbers and $n$ for variables.
\begin{align*}
e &::= i\ |\ e + e\ |\ e - e\ |\ e * e\ |\ n\\
\phi &::= e=e\ |\ e>e\ |\ \top\ |\ \bot\ |\ \neg \phi \ |\ \phi \lor \phi
\end{align*}

\noindent Although, in principle, our grammar allows arbitrary polynomial arithmetic expressions,
most of our examples are restricted to linear expressions.
We found that for any higher-degree polynomials, inference via Z3 becomes infeasible.

\vspace{-0.4em}
\subsection{Type Variables and Signatures}

We operate within a signature $\Sigma$ containing \emph{type definitions}, 
\emph{process declarations}, and \emph{process definitions} as follows:

\begin{align*}
\Sigma ::=\ &\cdot \mid \Sigma,V[\num{n} \mid \phi] = A &&\text{(type definition)}\\
\mid\ &\Sigma,\Delta \vdash f[\num{n}] :: (z: A) &&\text{(process declaration)}\\
\mid\ &\Sigma,\espawn{x}{f[\num{n}]}{\num{y}} = P  &&\text{(process definition)}\\
\Delta ::=\ &\cdot\ |\ \Delta, (x: A) &&\text{(process context)}
\end{align*}

The type definition $V[\num{n} \mid \phi] = A$ means that $V$ is indexed by some refinements $\num{n}$
such that $\phi$ holds; when we write $V[\num{e}]$, we interpret this as $A[\num{e}/\num{n}]$.
We adopt an \emph{equirecursive} and \emph{structural} approach to types, foregoing any explicit
communication for unfolding.
The process \emph{declaration} $\Delta \vdash f[\num{n}] :: (z: A)$ tells us that $f$ offers a
channel $z$ of type $A$, and that $f$ consumes all channels $x$ of type $A$ in $\Delta$.
The process $f$ is also indexed by some refinements $\num{n}$ which can be freely used in the types in
$\Delta$ and $A$.
Its corresponding \emph{definition} $\espawn{x}{f[\num{n}]}{\num{y}} = P$ is likewise indexed by $\num{n}$,
but instead of type information, gives us the process expression $P$ for $f$ using channels $x$ and $\num{y}$.
Again, variables in $\num{n}$ can appear freely in definitions, e.g., in assertions and assumptions.

\vspace{-1em}
\section{Subtyping}\label{sec:subtyping}

To define a subtyping algorithm and prove its correctness, we must first begin with a semantic definition
of subtyping, which is adopted from \citet{gaySubtypingSessionTypes2005}'s definition using a
\emph{type simulation}.

\vspace{-0.4em}
\paragraph*{\textbf{Type Simulation}}

Our definition of a type simulation generalizes prior work by \citet{das20}, who defined a
type \emph{bi}simulation for the purposes of type equality.
We rely upon the notion of \emph{closed} types:
a type is closed if it contains no free arithmetic variables.
To account for types with refinement variables, we introduce the notion of \emph{validity}.
Recall that the signature $\Sigma$ collects \emph{all} type definitions of the form $V[\num{n} \mid \phi]=A$.
This signature is then called \emph{valid} if
the implicit constraint $\phi$ holds for all occurrences $V[\num{e}]$ in the signature.
Formally, if $V[\num{e}]$ appears in any type definition in $\Sigma$, then $\vDash \phi[\num{e}/\num{n}]$.
We further require a valid signature to only contain \emph{contractive} type definitions,
disallowing definitions of the form $V[\num{n} \mid \phi] = V'[\num{e}]$.
We include a complete formal definition of valid types and signatures in Appendix~\ref{appendix-valid}.

\begin{definition}
On a valid signature $\Sigma$, we define $\m{unfold}_{\Sigma}(V[\num{e}])=A[\num{e}/\num{n}]$
if $V[\num{n} \mid \phi]= A \in \Sigma$ and $\m{unfold}_{\Sigma}(A)=A$ otherwise (when $A$
is not a type variable).
\end{definition}

A valid signature is essential for a well-defined $\m{unfold}_\Sigma$,
since the contractive condition prohibits us from unfolding to another type variable
(potentially indefinitely).

\begin{definition}[Type Simulation]
\label{tpsim}
A relation $\mathcal{R}$ on closed, valid types is a \emph{type simulation} under $\Sigma$ if, for any types $A, B$ such that $(A, B) \in \mathcal{R}$, when we take $S=\m{unfold}_\Sigma(A)$ and $T=\m{unfold}_\Sigma(B)$, the following holds:
\begin{enumerate}[i)]
\item If $S=\ichoice{\ell : A_\ell}_{\ell \in L}$, then $T=\ichoice{m : B_m}_{m \in M}$. Also, $L \subseteq M$
and $(A_\ell, B_\ell) \in \mathcal{R}$ for all $\ell \in L$.
\item If $S=\echoice{\ell : A_\ell}_{\ell \in L}$, then $T=\echoice{m : B_m}_{m \in M}$. Also, $L \supseteq M$ and $(A_m, B_m) \in \mathcal{R}$ for all $m \in M$.
\item If $S=A_1 \tensor A_2$, then $T=B_1 \tensor B_2$. Also, $(A_1, B_1) \in R$ and $(A_2, B_2) \in R$.
\item If $S=A_1 \lolli A_2$, then $T=B_1 \lolli B_2$. Also, $(B_1, A_1) \in R$ and $(A_2, B_2) \in R$.
\item If $S=\one$, then $T=\one$.
\item If $S=\texists{\phi}{A}$, then $T=\texists{\psi}{B}$. Also, either $\vDash \phi$, $\vDash \psi$, and $(A, B) \in \mathcal{R}$; or $\vDash \neg \phi$.
\item If $S=\tforall{\phi}{A}$, then $T=\tforall{\psi}{B'}$. Also, either $\vDash \psi$, $\vDash \phi$, and $(A, B) \in \mathcal{R}$; or $\vDash \neg \psi$.
\item If $S=\exists m . A$, then $T=\exists n . B$. Also, $\forall i \in \mathbb{N}$ we have $(A[i/m], B[i/n]) \in \mathcal{R}$.
\item If $S=\forall m . A$ then $T=\forall n . B$. Also, $\forall i \in \mathbb{N}$ we have $(A[i/m], B[i/n]) \in \mathcal{R}$.
\end{enumerate}
\end{definition}

Each type constructor has a corresponding case in Definition \ref{tpsim} which intuitively reduces
to the notion that the first type in a pair $A$ should \emph{simulate} the behavior of the second type $B$.
Concretely, the communication behaviors allowed by $A$ must be a subset of behaviors allowed by $B$.
For instance, in case $(i)$ for $\oplus$, we require that the label-set for $A$ is a subset of the
label-set for $B$ i.e. every label that a provider of $A$ can send can be received by a client of $B$
but not vice-versa.
Likewise, case $(iv)$ for $S=A_1 \lolli A_2 $ and $T=B_1 \lolli B_2$ notably adds
$(B_1, A_1)$ to $\mc{R}$ since $\lolli$ (like arrow types) is contravariant in the first argument.
Intuitively, if a process offers $(c : S)$ and we want to type it as $(c : T)$,
the assigned type $B_1$ of the received channel ought to exhibit a \emph{subset} of the behavior of the actual type $A_1$.
Cases $(i)-(v)$ have been borrowed from \citet{gaySubtypingSessionTypes2005}'s subtyping definition.

Worth noting are novel cases $(vi) - (ix)$ as they concern refinements.
Case $(vi)$ states that for $?\{\phi\}.\,A$, $\phi$ must either be true or false.
In the former case, $\psi$ must also be true and $(A, B) \in \mc{R}$.
However, in the latter case, the relation holds vacuously since the provider of such a
channel will cease to communicate, which effectively simulates \emph{any} continuation type.
Case $(vii)$ is analogous but flips the directionality of constraints.
Finally, case $(viii)$ (resp. $(ix)$) say that a quantified type $\exists m.\, A$
(resp. $\forall m.\, A$) simulates another one if all substitutions of $m$ are already
in the type simulation.
With type simulations in tow, we then define subtyping simply as follows:

\begin{definition} \label{defst}
For closed valid types $A$ and $B$, we say $A <: B$, i.e. $A$ is a \emph{subtype} of $B$, if there is a type simulation $\mathcal{R}$ such that $(A, B) \in \mathcal{R}$.
\end{definition}

\vspace{-0.4em}
\subsection{Algorithmic Subtyping}

Subtyping is fundamentally undecidable in the presence of refinements.
This fact is entailed by the fact that even the \emph{simpler} problem of type equality
is undecidable in the presence of refinements, as shown by prior work~\cite{das20}.
Nevertheless, we propose a sound algorithm which approximates subtyping, expressed via a
series of inference rules with a primary judgment of the form $\mc{V} \semi \mc{C} \semi \Gamma \Vdash A \st B$.
Here, $\mc{V}$ and $\mc{C}$ respectively represent the list of free arithmetic variables and
the governing constraint, and we likewise invoke the auxiliary judgment $\mc{V} \semi \mc{C} \vDash \phi$
to represent semantic entailment.
$\Gamma$ is a list of closures of the form
$\langle \mc{V} \semi \mc{C} \semi V_1[\num{e_1}] \st V_2[\num{e_2}] \rangle$ that have already
been encountered which we capture for the purposes of loop detection.
We omit some standard rules for brevity; the full series of subtyping rules can be found in appendix \ref{appendix-st}.

\begin{figure}[t]
\begin{mathpar}
\inferrule*[right=$\m{st}_?$] {
  \mathcal{V} \semi \mathcal{C} \vDash \phi \rightarrow \psi
  \and
  \mathcal{V} \semi \mathcal{C} \land \phi \semi \Gamma \Vdash A \st B
} {
  \mathcal{V} \semi \mathcal{C} \semi \Gamma \Vdash \texists{\phi}{A} \st \texists{\psi}{B}
}
\and
\inferrule*[right=$\m{st}_!$] {
  \mathcal{V} \semi \mathcal{C} \vDash \psi \rightarrow \phi
  \and
  \mathcal{V} \semi \mathcal{C} \land \psi \semi \Gamma \Vdash A \st B
} {
  \mathcal{V} \semi \mathcal{C} \semi \Gamma \Vdash \tforall{\phi}{A} \st \tforall{\psi}{B}
}
\and
\inferrule*[right=$\m{st}_{\exists}$] {
  \fresh{k}
  \and
  \mathcal{V}, k \semi \mathcal{C} \semi \Gamma \vdash A[k/m] \st B[k/n]
} {
  \mathcal{V} \semi \mathcal{C} \semi \Gamma \vdash \exists m . A \st \exists n . B
}
\and
\inferrule*[right=$\m{st}_\bot$] {
  \mathcal{V} \semi \mathcal{C} \vDash \bot
} {
  \mathcal{V} \semi \mathcal{C} \semi \Gamma \Vdash A <: B
}
\and
\inferrule*[right=$\m{st}_\m{expd}$] {
  V_1[\num{v_1} \vert \phi_1] = A \in \Sigma
  \and 
  V_2[\num{v_2} \vert \phi_2] = B \in \Sigma
  \and
  \gamma = \langle \mathcal{V} \semi \mathcal{C} \semi V_1[\num{e_1}] \st V_2[\num{e_2}] \rangle
  \\\\
  \mathcal{V} \semi \mathcal{C} \semi \Gamma, \gamma \Vdash A[\num{e_1}/\num{v_1}] \st B[\num{e_2}/\num{v_2}]
} {
  \mathcal{V} \semi \mathcal{C} \semi \Gamma \Vdash V_1[\num{e_1}] \st V_2[\num{e_2}]
}
\and
\inferrule*[right=$\m{st}_\m{def}$] {
  \langle \mathcal{V}' \semi \mathcal{C}' \semi V_1[\num{e_1}'] \st V_2[\num{e_2}'] \rangle \in \Gamma
  \and
  \mathcal{V} \semi \mathcal{C} \vDash \exists \mathcal{V}' . \mathcal{C}' \land \num{e_1}' = \num{e_1} \land \num{e_2}' = \num{e_2}
} {
  \mathcal{V} \semi \mathcal{C} \semi \Gamma \Vdash V_1[\num{e_1}] \st V_2[\num{e_2}]
}
\end{mathpar}
\vspace{-2em}
\caption{A selection of rules for the subtyping algorithm.}
\vspace{-1.5em}
\label{fig-stalg}
\end{figure}

Figure~\ref{fig-stalg} describes selected rules for subtyping concerning arithmetic refinements, where
the most interesting rules are $\m{st}_?$ and $\m{st}_!$.
The first premise of rule $\m{st}_?$ states that $\phi$ must imply $\psi$, capturing the intuition
of Definition~\ref{tpsim} that either $\phi$ and $\psi$ both hold or $\phi$ is false.
Note that we only say $\phi \rightarrow \psi$ \emph{under some constraint} $\mc{C}$,
but that in Definition~\ref{tpsim} we don't have any such constraint:
this is because type simulations are only defined for closed types.
(We will see in the proof how to bridge this gap.)
The second premise requires that continuation type $A$ must be a subtype of $B$ under the constraint
$\mc{C} \land \phi$.
Rule $\m{st}_!$ is analogous, only flipping the direction of implication, similar to Definition~\ref{tpsim}.
Rule $\m{st}_\exists$ simply introduces a fresh variable $k$ to our variable context $\mc{V}$,
but also requires that we substitute $k$ for the variable specified in each type ($m$ or $n$).
Rule $\m{st}_{\bot}$ handles the cases where the constraint $\mc{C}$ is contradictory.
Under such a constraint, arbitrary types $A$ and $B$ are subtypes since such a situation will never
arise at runtime.
This rule comes in handy when type checking branches that are impossible due to refinements (e.g.
$\mb{zero}$ branch on type $\m{nat}[n+1]$).

The heart of loop detection lies in the $\m{st}_\m{expd}$ and $\m{st}_\m{def}$ rules,
each of which may apply when we encounter type names $V_1[\num{e_1}]$ and $V_2[\num{e_2}]$.
$\m{st}_\m{expd}$ effectively unrolls those type names according to our signature $\Sigma$,
but it also generates a closure $\gamma$ and adds it to our context $\Gamma$,
denoting that we have ``seen'' $V_1$ and $V_2$, refined by these particular expressions
$\num{e_1}$ and $\num{e_2}$, under this particular $\mc{V}$ and $\mc{C}$.
Now if, later in our algorithm, we encounter those same $V_1$ and $V_2$ again,
we eagerly follow the $\m{st}_\m{def}$ rule.
We do not require that our new $\mc{V}'$ and $\mc{C}'$ 
are identical to the $\mc{V}$ and $\mc{C}$ in $\gamma$—
the second premise asserts (informally) that we can find a substitution between them,
and if such a substitution does exist, we conclude that our algorithm succeeds.
This coinductive structure allows our algorithm to handle types which could otherwise be expanded infinitely.

$\m{st}_{\m{expd}}$ requires that both types are unfolded simultaneously,
and so as written, we cannot compare a name to a structure. To circumvent this issue in practice,
we employ a preprocessing step called \emph{internal renaming} \citep{dasNested20}.
For every type definition $V[n]=A$ in our signature, we rename type expression $A$ by
replacing inner type expressions with an explicit fresh name.
For instance, we rewrite the definition for
$\m{list_A} = \ichoice{\mb{nil}: \one,\ \mb{cons}: A \tensor \m{list_A}}$
 of some type (name) $A$ into the following set of definitions:

\vspace{-1.5em}
\begin{align*}
  \m{type}\ \m{list_A} &= \ichoice{\mb{nil}: \one,\ \mb{cons}: \m{\%0}} &
  \m{type}\ \m{\%0} &= A \tensor \m{list_A}
\end{align*}

\noindent where $\%0$ is a fresh name.
We apply this renaming recursively inside type expressions in a way that a type definition
always has the form $A \; \m{op} \; B$ where $\m{op}$ is a type operator.
This way, type names and structural types always \emph{alternate}, in any type expression.
This guarantees that, while applying subtyping algorithmic rules, we also alternate between
comparing names against names, and structural types against structural types.

These $\m{st}_{\m{expd}}$ and $\m{st}_{\m{def}}$ rules are also noteworthy for introducing incompleteness to our algorithm.
On encountering two type names, it is possible to keep applying $\m{st}_{\m{expd}}$ indefinitely,
which would lead to a non-terminating algorithm. To circumvent this problem, we only expand a finite number of times,
and then forcibly apply the $\m{st}_{\m{def}}$ rule. However, this comes at the cost of completeness, since
$\m{st}_{\m{def}}$ may then fail to find a matching closure in $\Gamma$.
In the absence of refinements, this is the only source of incompleteness in subtyping.

\vspace{-0.4em}
\subsection{Soundness} \label{sec:soundness}

We now prove that our algorithm is sound with respect to Definition \ref{defst}.
Proving subtyping of $A \st B$ effectively reduces to \emph{constructing} a type simulation
$\mc{R}$ such that $(A, B) \in \mc{R}$.
The main intuition here is that such a type simulation can be constructed from the subtyping
derivation built using our algorithmic rules.
We here motivate the key ideas; a full proof sketch can be found in Appendix \ref{appendix-proof}.

\begin{definition} \label{paper:forallvstc}
  We define $\forall \mathcal{V} . \mathcal{C} \rightarrow A \st B$ (read: for all $\mc{V}$, $\mc{C}$
  implies $A$ is a subtype of $B$) if there exists a type simulation
  $\mathcal{R}$ such that, for all ground substitutions $\sigma$ over $\mc{V}$ such that
  $\vDash \mathcal{C}[\sigma]$, we get $(A[\sigma], B[\sigma]) \in \mathcal{R}$.
\end{definition}
This definition extends the idea of subtyping from closed types (e.g., $\m{nat}[2]$) to
symbolic types (e.g., $\m{nat}[n]$).
This is necessary since our subtyping algorithm works on symbolic types.
The intuition behind this definition is that $A$ is a subtype of $B$ if for all substitutions $\sigma$
that satisfy $\mc{C}$, $A[\sigma]$ is a subtype of $B[\sigma]$.
As an example, we can say $\forall n. \, n > 0 \rightarrow \m{nat}[n] \st \m{nat}[n]$
because $\m{nat}[1] \st \m{nat}[1]$, $\m{nat}[2] \st \m{nat}[2]$, and so on.

\begin{definition}
Given a derivation $\mathcal{D}$ of $\mathcal{V} \semi \mathcal{C} \semi \Gamma \Vdash A \st B$, we define the set of closures $S(\mathcal{D})$ such that, for each sequent $\mathcal{V}' \semi \mathcal{C}' \semi \Gamma' \Vdash A' \st B'$, we include $\langle \mathcal{V}' \semi \mathcal{C}' \semi A' \st B' \rangle \in S(\mathcal{D})$.
\end{definition}

This set of closures from a derivation is exactly what we need to construct the type simulation.
The key idea here is that for a well-formed derivation, all these closures must represent valid subtyping
relations as well.
Finally, the type simulation is constructed by applying ground substitutions to all closures
found in the derivation.
This intuition is captured in the following theorem.

\begin{theorem}[Soundness]\label{thm:sound}
If $\mathcal{V} \semi \mathcal{C} \semi \cdot \Vdash A \st B$, then $\forall \mathcal{V} . C \rightarrow A \st B$.
\end{theorem}

\begin{proof}
By taking the derivation $\mc{D}$ as above, and constructing the type simulation 
\[\mathcal{R} = \{(A'[\sigma], B'[\sigma])\ \vert\ \langle \mathcal{V'} \semi \mathcal{C'} \semi A' \st B' \rangle \in S(\mathcal{D}) \text{ and } \sigma : \mc{V'} \text { and }\vDash \mc{C'}[\sigma]\}\]
We show in Appendix \ref{appendix-proof} that $\mc{R}$ is indeed a type simulation.
\end{proof}

\vspace{-1em}
\section{Inference Algorithm}\label{sec:inference}

This section presents our rules for constraint generation, which form the theoretical backbone of our type inference algorithm.
The rules themselves are in close correspondence with the standard typechecking rules for our language, included in Appendix \ref{appendix-tc}.
Our primary judgment takes the same form $\mc{V} \semi \mc{C} \semi \Delta \vdash P :: (z: C)$ 
for a program $P$ offering channel $z$ of type $C$, and consuming all channels $(x:A_x) \in \Delta$. 
As with subtyping, we also operate within the context of a list of free arithmetic variables $\mc{V}$ a
nd a governing constraint $\mc{C}$. For comparison, we include below the typechecking rule for $\oplus R$, read from bottom to top:
\begin{mathpar}
\inferrule*[right=tc-$\oplus$R] {
  (k \in L)
  \and
  \mathcal{V} \semi \mathcal{C} \semi \Delta \vdash P :: (x: A_k)
} {
  \mathcal{V} \semi \mathcal{C} \semi \Delta \vdash x.k \semi P :: (x: \ichoice{\ell: A_\ell}_{\ell \in L})
}
\end{mathpar}
Since we are doing inference as opposed to typechecking, we amend this rule to reflect that we are not given $x$'s type structure \emph{a priori}.
We treat it instead as a variable to which we apply a \emph{type constraint}, which conveniently takes the form of the subtyping judgment $\mc{V} \semi \mc{C} \Vdash A \st B$ for types $A$ and $B$ asserting that $A$ is a subtype of $B$. We highlight the introduced intermediate type names and new arithmetic variables in blue, and assume all such names are fresh.
\begin{mathpar}
\inferrule*[right=$\oplus$R] {
  \mc{V} \semi \mc{C} \Vdash \ichoice{k: \textcolor{blue}{B}} \st A
  \and
  \mathcal{V} \semi \mathcal{C} \semi \Delta \vdash P :: (x: \textcolor{blue}{B})
} {
  \mathcal{V} \semi \mathcal{C} \semi \Delta \vdash x.k \semi P :: (x: A)
}
\end{mathpar}
Our resulting $\oplus$R rule, again read bottom-to-top, is a simple structural rule which generates a single subtyping constraint.
Informally, when we encounter the program $\Delta \vdash x.\mathsf{succ} \semi P :: (x: A)$, we know that $A$ must be a $\oplus$-type, and that it must contain the label $\m{succ}$: both pieces of information are captured via the type constraint premise. We would then continue our algorithm on the latter premise with the freshly
introduced intermediate type $B$.
The other structural rules are analogous and omitted for brevity.

The above rule (and all structural rules) only create type constraints.
In addition, some rules yield constraints of another sort: \emph{arithmetic constraints},
which take the form of semantic entailment $\mc{V} \semi \mc{C} \vDash \phi$ for a proposition $\phi$, asserting that $\phi$ holds under $\mc{V}$ and $\mc{C}$.
These rules are generated for the refinement type constructors that are responsible for
exchanging proofs.
In what follows, we highlight a few representative rules. A full list can be found in Appendix \ref{appendix-cgen}.
\begin{mathpar}
\inferrule*[right=?R] {
  \mathcal{V} \semi \mathcal{C} \vDash \phi
  \and
  \mathcal{V} \semi \mathcal{C} \Vdash \texists{\phi}{\textcolor{blue}{A'}} \st A
  \and
  \mathcal{V} \semi \mathcal{C} \semi \Delta \vdash P :: (x: \textcolor{blue}{A'})
} {
  \mathcal{V} \semi \mathcal{C} \semi \Delta \vdash \eassert{x}{\phi} \semi P :: (x: A)
}
\end{mathpar}
The ?R rule is notable for introducing both an arithmetic constraint and a type constraint. The arithmetic constraint mirrors that of the corresponding typechecking rule, and simply states that $\phi$ must hold. Our type constraint, in tandem with the subtyping rules, asserts that $A$ must have the structure $\texists{\psi}{A'}$ for some $\psi$, but $\psi$ need not be identical to $\phi$: we instead must have $\mc{V} \semi \mc{C} \vDash \phi \rightarrow \psi$, i.e. $\phi$ is \emph{stronger} than $\psi$.
\begin{mathpar}
\inferrule*[right=?L] {
  \mathcal{V} \semi \mathcal{C} \Vdash A \st \texists{\phi}{\textcolor{blue}{A'}}
  \and
  \mathcal{V} \semi \mathcal{C} \land \phi \semi \Delta, (x: \textcolor{blue}{A'}) \vdash Q :: (z: C)
} {
  \mathcal{V} \semi \mathcal{C} \semi \Delta, (x: A) \vdash \eassume{x}{\phi} \semi Q :: (z: C)
} 
\end{mathpar}
If the channel of the same type appears in our context $\Delta$, we instead apply the ?L rule. In right rules, we see that the freshly generated type appears to the left of the subtype constraint.
In contrast, for all left rules, the fresh type falls on the right due to the duality of types.
When we offer a channel $(z: C)$, we want that our declared type is \emph{broader}, or a supertype, of whatever our program necessitates— 
e.g. we might declare $(z: \m{nat})$ even if our process only sends the $\m{zero}$ label. 
Conversely, if we consume a channel $(x: A)$, we need $A$ to be \emph{narrower}, or a subtype, of what our program dictates, 
i.e. the program should be able to handle any of $A$'s behavior. 
In this instance where $A=\texists{\psi}{A'}$ for some $\phi$, subtyping dictates that 
$\psi$ should be at least as strong as $\phi$, since if $(x: A)$ "sends" a proof of $\psi$, we consequently receive a proof of $\phi$.
\begin{mathpar}
\inferrule*[right=$\exists$R] {
  \mathcal{V} \semi \mathcal{C} \vDash e \geq 0
  \and
  \mathcal{V} \semi \mathcal{C} \Vdash \exists n.\textcolor{blue}{A'} \st A
  \and
  \mathcal{V} \semi \mathcal{C} \semi \Delta \vdash P :: (x: \textcolor{blue}{A'})
} {
  \mathcal{V} \semi \mathcal{C} \land n=e \semi \Delta \vdash \esend{x}{\{e\}} \semi P :: (x: A)
}
\end{mathpar}
The $\exists$R rule again contains both an arithmetic and a type constraint. The arithmetic constraint dictates that whatever expression we send should be a natural number, i.e., $e \geq 0$ to maintain
the invariant that all witnesses are natural numbers.
Of special note in this particular rule is that, since we are doing inference, we cannot perform a substitution $A'[e/n]$ as in typechecking, since we lack any information about where $n$ might occur in the fresh name $A'$. As a tidy solution, we instead store the substitution as an equivalence in our governing constraint $\mc{C}$, such that any further continuations have access to the value of $n$ without our needing to know ahead of time.
\begin{mathpar}
\inferrule*[right=$\m{def}$] {
  (\num{y_i' : B_i'})_{i \in I} \vdash f[\num{n}\: |\: \phi] = P_f :: (x' : A') \in \Sigma
  \and\\
  \mathcal{V} \semi \mathcal{C} \land \phi[\num{e}/\num{n}] \Vdash A'[\num{e}/\num{n}] \st  \textcolor{blue}{A}
  \and
  (i \in I) \quad \mathcal{V} \semi \mathcal{C} \land \phi[\num{e}/\num{n}] \Vdash B_i \st B_i' [\num{e}/\num{n}]
  \\
  \mathcal{V} \semi \mathcal{C} \semi \Delta, (x:  \textcolor{blue}{A}[\num{e}/\num{n}]) \vdash Q :: (z: C)
} {
  \mathcal{V} \semi \mathcal{C} \semi \Delta, (\num{y_i : B_i})_{i \in I} \vdash \espawn{x}{f[\num{e}]}{\num{y}} \semi Q :: (z: C)
}
\end{mathpar}
Our final example, the $\m{def}$ rule for spawning processes, is noteworthy for interacting with the signature: our declaration of $f$ has $A'$ and $B_i'$ as type variables, which are constrained both by the body $P_f$ of $f$ and here by our spawning of $f$. Our subtyping constraints are deceptively straightforward: for each channel $(y_i: B_i)$ consumed by our spawning of $f$, we compare it with the corresponding channel in the declaration $(y_i': B_i')$ of $f$, and we assert that the "real" type $B_i$ is a subtype of the "expected" type $B_i'$. Conversely, we also constrain the offered type $A'$ of $f$ such that $A'$ is a subtype of our fresh variable $A$. 

In practice, the $\m{def}$ rule is especially important in that it generates the majority of our arithmetic constraints, \emph{despite not generating any as an explicit premise}. To solve type constraints, we apply our subtyping algorithm as presented in Section~\ref{sec:subtyping}, which consequently generates arithmetic constraints. We elaborate further on the relationship between type and arithmetic constraints, as well as the constraint solving process as a whole, in the following section.

We conclude with a statement on the soundness of our inference algorithm:

\begin{theorem}\label{thm:algocorrect}
Let $z \leftarrow f[\mc{V}]\ \num{x_i} = P$ be the definition of a process $f$. 
For fresh type names $A_i$ and $C$ such that $(z: C)$ and $\Delta = (x_i: A_i)_{\forall i}$, 
and if $\mc{V} \semi \top \semi \Delta \vdash P :: (z: C)$, then $P$ is a well-typed process,
i.e. the type assignment for $A_i$ and $C$ produced by inference satisfies typechecking.
\end{theorem}

\vspace{-1em}
\section{Implementation}\label{sec:impl}

We implemented the type inference algorithm on top of the Rast implementation~\cite{Das20FSCD},
which already supports lexing, parsing, and typechecking for session types with arithmetic refinements.
The inference engine consists of 2,069 lines of SML code, and integrates with the
Z3 SMT solver~\cite{Moura08Z3} to model arithmetic constraints. 
Our general approach begins by introducing placeholder declarations for each process in the file;
for every definition \lst{proc x <- f[$\mc{V}$] x1 x2 ... xn}, we introduce the declaration
\begin{lstlisting}[language=prest]
	decl f[$\mc{V}$] : (x1 : $A_1[\overline{e_1}]$) (x2 : $A_2[\overline{e_2}]$) ... (xn : $A_n[\overline{e_n}]$) |- (x : $A[\overline{e}]$)
\end{lstlisting}
where each type $A_1, A_2, \ldots, A_n, A$ is a fresh type variable refined with
arithmetic expressions $e_i(\mc{V})$ over the free variables $\mc{V}$ of the process.
Our goal is to find a \emph{type assignment} which maps each type variable to a concrete type name
declared in the program,
as well as an \emph{expression assignment} which maps each $e_i$ to a symbolic arithmetic expression.
Theorem~\ref{thm:algocorrect} dictates that these assignments result in a well-typed program,
which is validated by running the Rast type checker on the inferred types.
We discuss how this inference is carried out in two stages, the practical challenges we
faced, and the optimizations we introduced to address them.

\paragraph*{\textbf{Two-Stage Inference}}
When inferring the type of a process, we form constraints of two different sorts:
type constraints and arithmetic constraints.
Our constraint generation rules on processes create both sorts, and our subtyping algorithm,
which we use to solve type constraints, may yield additional arithmetic constraints
that must also be satisfied for inference to succeed.

At first, it might be tempting to solve both sorts of contraints in a unified pass.
However, this becomes infeasible in practice.
To see why, first note that arithmetic constraints cannot be solved \emph{eagerly} due to
the presence of process spawns.
Consider a process \lst{foo} that calls the \lst{two} process twice, assigning the
offered channel to \lst{y} and \lst{z} respectively.
\begin{lstlisting}[language=prest,mathescape]
decl foo: . |- (x: $\rt{A}{e()}$)
proc x <- foo = y <- two; z <- two; ...
\end{lstlisting}
Both calls to \lst{two} in the \lst{foo} process will generate constraints on type $\rt{A}{e()}$,
which in turn will generate multiple distinct arithmetic constraints on $e$.
Since $e$ must be the same in \emph{both} constraints, it is insufficient to demonstrate
the satisfiability of these two constraints independently; a solver might otherwise invalidly
assign different values to $e$.
Thus, we must first collect all arithmetic constraints and then solve them all in one large batch.

However, taking this approach and otherwise following the rules as written leads to
remarkable blowup in the number of constraints we generate and expensive solver calls we must make.
The primary reason for this blowup is the presence of the $\m{st_\bot}$ rule,
which dictates that arbitrary, even mismatched, types can be subtypes
under constraint $\mc{C}$ as long as $\mc{C}$ is false.
Thus, whenever we might be tempted to return false when subtyping on a structural mismatch,
such as $\mc{V} \semi \mc{C} \vDash A \st B$, we must instead add to our arithmetic constraint
set the constraint $\mc{V} \semi \mc{C} \vDash \bot$.
Thus, in this na\"{i}ve approach, we are not able to rule out \emph{any} type assignments
without making a solver call, which becomes prohibitively expensive due to the exponential
number of possible assignments.

We therefore chose to restrict the programmer from writing any dead code; that is,
code where $\mc{V} \semi \mc{C} \vDash \bot$.
This restriction has a number of consequences, but the essential one is that we can now
rule out type assignments on structure without checking whether $\mc{C}$ is satisfiable.
Taking advantage of this restriction, we split type inference into two stages:
the first stage operates solely upon type structures, where all refinements are stripped
from the program, and yields a list of structurally viable type assignments;
the second stage reconsiders the refinements for a specific type assignment, collecting the resulting arithmetic constraints and passing them into the Z3 SMT solver.
This two-stage approach has the additional benefit of drawing a clean boundary
between the decidable and undecidable fragments of our algorithm:
in particular, type inference \emph{without} refinements is decidable,
and so we can return determinate error messages upon encountering a type mismatch.

\vspace{-0.4em}
\subsection{First Stage of Inference}
The first stage collects the constraints generated during type checking and subtyping,
strips off the refinements, and proceeds to compute a type assignment.
Since this stage is actually decidable, our algorithm is guaranteed to terminate and
either find a satisfying type assignment, or report a type error.
Thus, if the program is ill-typed due to structural typing (i.e., without refinements),
then it is guaranteed to be detected.

\paragraph*{\textbf{Transitivity}}

As described, our type constraint generation rules produce a fresh ``intermediate'' type variable for every process expression we encounter. 
For instance, consider the $\mt{double}$ proces which doubles the value of a given natural number. 
Our rules will generate the following constraints, denoted with comments on the line that induces the constraint:
\begin{lstlisting}[language=prest]
decl double: (x: A0) |- (y: A1)
proc y <- double x = case x (            % A0 <: +{zero: I0, succ: I1}
    zero => y.zero;                     % +{zero: I2} <: A1
             y <-> x                      % I0 <: I2
  | succ => y.succ;                     % +{succ: I3} <: A1
             y.succ;                     % +{succ: I4} <: I3
             y <- double x )             % I1 <: A0; A1 <: I4
\end{lstlisting}
We generate fresh intermediate type variables (denoted by $I_n$) as a placeholder for
the continuation type.
If we do not distinguish between these intermediate variables and the ``signature'' variables (e.g., $A$)
which we introduce in our declarations, we will generate a large number of constraints, putting a heavy
load on our inference engine.
To make matters worse, when these types have index refinements, each such intermediate variable will also
create arithmetic constraints, which will make constraint solving infeasible for Z3.
This imposes a serious limitation on the scalability of inference.

We address this limitation with a crucial observation:
\emph{if $A \st B$, then so is $\ichoice{k : A} \st \ichoice{k : B}$}.
To see this in action in the above example, note the third constraint:
$I_0 \st I_2$ from which we deduce that $\ichoice{\mb{zero} : I_0} \st \ichoice{\mb{zero} : I_2}$.
Since our second constraint asserts that $\ichoice{\mb{zero} : I_2} \st A_1$,
we transitively conclude that $\ichoice{\mb{zero} : I_0} \st A_1$,
and we have thus eliminated $A_2$ from our constraints.
Following this process, we can reduce the above set of seven constraints
to just two: $A_0 <: \ichoice{\mb{zero}: I_0, \mb{succ}: A_0}$ 
and $\ichoice{\mb{zero}: I_0, \mb{succ}: \ichoice{\mb{succ}: A_1}} <: A_1$.
Here our two constraints correspond to the most general types of $A_0$ and $A_1$.

This observation leads to the most significant optimization in our inference engine,
which we call \emph{transitivity}. It partly relies on transitivity of subtyping,
the property that if $A <: B$ and $B <: C$, then $A <: C$.
Before we generate a type assignment, we perform a transitivity pass, which eliminates
intermediate variables whenever possible.
This dramatically reduces the number of type and arithmetic constraints, leading
to significant speedups in inference, which we evaluate in Section~\ref{sec:eval}.
As in the example, not all intermediate variables will be removed by transitivity though---specifically,
forwarding between channels which have both previously been communicated on
will result in a type constraint of the form $\mc{V} \semi \mc{C} \vDash I_0 \st I_1$.
In such cases, we choose $I_1$ to perform any appropriate substitutions by transitivity,
and we leave $I_0$ "free", that is, unconstrained—$I_0$ will only appear in the
continuation type of another constraint.

\vspace{-0.5em}
\paragraph*{\textbf{Type Constraint Solution}} Once the intermediate variables are eliminated, we assign, to each remaining type variable,
an initial \emph{search space} consisting of each defined type name in the file.
Then, for each type variable, we substitute the type variable with one type from its search space,
and run a \emph{partial} version of subtyping on any relevant constraints.
This modified subtyping algorithm stops eagerly upon encountering two (non-identical) type names,
and returns either a reduced constraint or simply true or false.
Failed candidates are removed from the search space while successful ones remain,
and might get pruned away due to refinements.

Continuing the previous example, suppose we have the following type definitions in our signature.

\begin{lstlisting}[language=prest]
type nat = +{zero: one, succ: nat}
type odd = +{succ: even}
type even = +{zero: one, succ: odd}
\end{lstlisting}

Recall our post-transitivity constraints $A_0 <: \ichoice{\mb{zero}: I_0, \mb{succ}: A_0}$ 
and $\ichoice{\mb{zero}: I_0, \mb{succ}: \ichoice{\mb{succ}: A_1}} <: A_1$.
Type constraint solution would determine that $A_0$ can be any one of these,
but that $A_1$ can only be either $\m{nat}$ or $\m{even}$, since substituting $\m{odd}$ for $A_1$
violates the second constraint by the $\m{st}_{\oplus}$ rule. 
Note that we can only eliminate $\m{odd}$ here because we assume there is no dead code:
otherwise, there might be some refinements that make the line $y.\mathsf{zero}$ unreachable,
letting us satisfy the constraint by applying $\m{st}_{\bot}$ instead of $\m{st}_{\oplus}$.


\vspace{-0.5em}
\paragraph*{\textbf{Program Reconstruction}}
The refinement-agnostic first stage also allows us to \emph{reconstruct} each process definition
according to its assigned type, adding back in assertions and assumptions.
We take the approach of \emph{forcing calculus}~\cite{dasVerified2020}, eagerly assuming whenever
possible and lazily asserting when absolutely necessary, guaranteeing that any assertions
``see'' all possible assumptions.
Reconstruction also means that the programmer no longer has to manually write assertions or assumptions—
a labor-intensive process which significantly bloats the source code, offsetting many potential
practical benefits of type inference.
We also reconstruct impossible branches, in line with our restriction that the programmer cannot
themselves write unreachable code.
Specifically, if a label in the type of a label-set is absent from a case statement,
we manually insert it and say that it is impossible.
This approach, in tandem with the two-stage solution,
constitutes our solution to the aforementioned impossible-blowup problem:
in the type solution stage, we assume that $\mc{C}$ is never false,
which permits us to eliminate possible type assignments on structure.

To conclude the example, we choose an arbitrary permutation of valid type assignments $A_0=\m{nat} \semi A_1=\m{even}$.
Then, supposing $\m{nat}$ employs the standard refinement, and $\m{odd}$ and $\m{even}$ follow suit,
we reconstruct the program refinements to conclude the first stage:

\begin{lstlisting}[language=prest]
decl double[n]: (x: nat[e0(n)]) |- (y: even[e1(n)])
proc y <- double x = case x (
    zero => assume x {e0(n) = 0};
        y.zero; assert y {e1(n) = 0};
        y <-> x
  | succ => assume x {e0(n) > 0};
        y.succ; assert y {e1(n) > 0};
        y.succ; assert y {e1(n)-1 > 0};
        y <- double x )
\end{lstlisting}






\vspace{-0.4em}
\subsection{Second Stage of Inference}



In the second stage, we collect the arithmetic constraints and attempt to find a satisfying assignment.
Since our first stage has already provided us with structurally sound type assignments,
we generally skip over structural constructs which are not refinement-related.
However, we pay extra attention to the judgmental constructs $\m{forward}$ and $\m{spawn}$,
since they imply subtyping relations between their operands.
For these, we run the complete subtyping algorithm and collect any resulting constraints.
This generation substage yields a list of arithmetic constraints which are shipped off to Z3~\cite{Moura08Z3},
which, if they are satisfiable, returns values which we can parse into an expression assignment.
If Z3 fails to find a solution, i.e., either because there is none, or the problem is undecidable,
or due to a timeout, we say there is no solution and move on to the next type assignment if there
exists one.
Note that this stage actually attempts to solve an undecidable problem, hence our algorithm
is incomplete.
However, because of our numerous optimizations and heuristics, it works remarkably well
in practice.

\paragraph*{\textbf{Polynomial Templates}}

The bulk of our implementation effort went towards coercing Z3 into cooperating with our style
of constraints and scaling it to a wide variety of challenging benchmarks.
We translate our arithmetic constraints of the form $\mc{V} \semi \mc{C} \vDash \phi$ into the logical formulas $\forall \mc{V}.\, \mc{C} {\implies} \phi$, where $\mc{C}$ and $\phi$
include expression variables $e_i(\mc{V})$ which we want to solve for.
Our initial approach was to treat $e_i$'s as \emph{uninterpreted functions} and rely solely upon Z3's
built-in uninterpreted function solver.
Unfortunately, this approach turned out to be insufficient for even the most trivial examples: even when Z3 did not time out, it would return a function interpretation which was piecewise or otherwise inexpressible in our language for arithmetic expressions.

Our next optimization is based on the observation that we only allow refinement expressions that are
\emph{polynomials} over the quantified variables.
Taking advantage of this, we represent each expression as a multivariate polynomial
of degree $d$, and solve for all introduced coefficients.
For instance, if $d=2$ and we have free variables $m, n$ for an expression $e$, we say
$e(m,n)=c_0+c_1m+c_2n+c_3m^2+c_4mn+c_5n^2$.
By providing this template to the solver, we greatly restrict the number of possible
interpretations which Z3 must consider relative to purely uninterpreted functions,
thereby reducing the solver burden significantly.
Although we can support non-linear refinements through this approach in principle,
we frequently encountered Z3 timeouts when solving non-linear constraints,
and so we chose to restrict the degree to $d=1$ in practice.


\paragraph*{\textbf{Real Arithmetic}}

Our refinements only allow natural numbers, and do not support real values in arithmetic expressions.
However, we found that Z3 stalls less frequently if we solve within a real-valued logic,
as opposed to an integer-valued one.
Of course, with this approach we risk Z3 returning non-integer values when modeling our coefficients,
which must ultimately be integers, but this occurrence turns out to be much less frequent than expected. 
Oftentimes, Z3 simply returns exact integer values, even when the same constraints would time out for
an integer logic. If Z3 returns a floating-point value, we fall back to integer logic.

\vspace{-1em}
\section{Evaluation}\label{sec:eval}

\paragraph*{\textbf{Methodology}}

We evaluate the performance and efficacy of our inference algorithm on a variety of
challenging benchmarks, that are known to be well-typed.
Experiments were performed on a 2021 MacBook Pro with 16GB of RAM and an 8-core M1 Pro CPU.
For each benchmark, we present the execution time of both stages: type constraint solving (Stage 1)
and arithmetic constraint solving (Stage 2).
We evaluate the efficacy of our three key optimizations as follows:
\begin{itemize}
	\item \emph{Polynomial Templates}: We run our inference engine with two strategies:
	\textbf{uif} stands for uninterpreted functions, while \textbf{poly} stands for
	using polynomial templates.

	\item \emph{Real Arithmetic}: Our calls to Z3 also have two modes:
	\textbf{real} stands for real arithmetic, while \textbf{int} stands for integers.

	\item \emph{Transitivity}: We enable/disable transitivity which eliminates the
	intermediate type variables, so that we do not need to find a satisfying type
	assignment for them.
\end{itemize}
Our results are summarized in Table~\ref{table-results}.
For all experiments, each call to Z3 was set to timeout after 10 seconds, and each trial,
which might make multiple such calls, was run for a maximum of 60 seconds in total.
All numerical results are averaged across 10 trials, rounded to the nearest hundredth of a millisecond.

Our algorithm yields four possible results:
\begin{itemize}
	\item \textbf{success}: inference finds a valid type and arithmetic assignment,
	\item \textbf{timeout}: either a call to Z3 for a particular type assignment fails to
	return within the specified 10-second
	limit, or the overall algorithm execution time exceeds the 60-second limit,
	\item \textbf{inexpressible}: when Z3, under \textbf{uif} mode, returns an expression assignment
	that cannot be expressed by our arithmetic language, e.g., expressions contain if-then-else constructs;
	or in the real arithmetic case, a floating-point number is returned that cannot be converted
	into an integer.
\end{itemize}
Technically, there is a fourth possibility as well: if the program is ill-typed,
our inference algorithm returns \textbf{unsat}.
However, for our evaluation, we chose to focus on well-typed benchmarks.
Additionally, we elide from Table~\ref{table-results} those experiments which resulted in a timeout.

\vspace{-0.4em}
\subsection{Results}

For each benchmark, we briefly explain its contents, present the experimental results,
and discuss any notable findings.

\paragraph*{\textbf{Unary Natural Numbers} (nat-u)}
This benchmark primarily contains the refined natural number based on the types
\begin{align*}
	& \mi{type}\ \m{nat}[n] = \ichoice{\mb{zero}: \texists{n=0}{\one},\ \mb{succ}: \texists{n>0}{\m{nat}[n-1]}} \\
	& \mi{type}\ \m{natpair}[m][n] = \m{nat}[m] \tensor \m{nat}[n] \tensor \one
\end{align*}

\noindent The $\m{natpair}$ type is necessary for duplicating numbers due to linearity restrictions.
The module includes four main processes and three helper and test processes:
\textbf{add} for adding two numbers,
\textbf{clone} for copying a number,
\textbf{consume} that converts a $\m{nat}$ into $\one$,
and \textbf{double} which doubles a number via clone and add.

First, we note that inference succeeds only when we use both transitivity
and polynomial templates.
Real arithmetic leads to a slowdown compared to integer arithmetic in Stage 2,
with no significant difference in Stage 1.
As was expected, Stage 2 takes 2-3 orders of magnitude time more than Stage 1,
since it involves calls to Z3.
The types inferred in each case were as expected, for example
\begin{lstlisting}[language=prest]
	decl add[m][n] : (x : nat[m]) (y : nat[n]) |- (z : nat[m+n])
	decl double[n] : (x : nat[n]) |- (y : nat[2*n])
\end{lstlisting}

\begin{table}[t]
\setlength{\tabcolsep}{8pt}
\begin{adjustbox}{width=\textwidth}
\begin{tabular}{c | c | c | c | c | r | r}
\textbf{Benchmark} & \textbf{Strategy} & \textbf{Arithmetic} & \textbf{Transitivity} & \textbf{Result} & \textbf{Stage 1} & \textbf{Stage 2} \\\hline
nat-u & poly & real & true & success & 0.54 & 613.64\\
nat-u & poly & int & true & success & 0.56 & 171.04\\\hline
nat-d & poly & real & true & success & 0.52 & 637.88\\
nat-d & poly & real & false & success & 17.07 & 665.02\\
nat-d & uif & real & true & inexpressible & 0.47 & 214.54\\
nat-d & uif & real & false & inexpressible & 16.88 & 249.26\\
nat-d & uif & int & true & inexpressible & 0.49 & 210.77\\
nat-d & uif & int & false & inexpressible & 16.87 & 246.57\\\hline
bin & poly & real & true & success & 1.26 & 255.89\\\hline
list & poly & real & true & success & 0.84 & 465.14\\
list & poly & int & true & success & 0.85 & 228.47\\\hline
linlam & poly & real & true & success & 1.39 & 11.00\\
linlam & poly & int & true & success & 1.39 & 10.56\\
linlam & uif & real & true & success & 1.40 & 9.84\\
linlam & uif & int & true & success & 1.39 & 9.63\\\hline
linlam-size & poly & real & true & success & 1.75 & 148.33\\
linlam-size & poly & int & true & success & 1.65 & 1091.58\\
\end{tabular}
\end{adjustbox}
\caption{Evaluation results.}
\vspace{-2.5em}
\label{table-results}
\end{table}

\paragraph*{\textbf{Direct Natural Numbers} (nat-d)}
An alternative representation for natural numbers is directly through refinements,
via the type $\m{nat} = \exists n. \one$.
Instead of sending `n' $\mb{succ}$ messages, this type produces a single natural number
as a refinement and terminates.
This benchmark contains all the same programs as unary naturals, but with this modified
representation.

Successful satisfying assignments only occur for the specific combination of polynomial templates
along with real arithmetic, regardless of transitivity.
Transitivity leads to a speedup of 1-2 orders of magnitude in Stage 1.
With uninterpreted functions, inference does find a solution but heavily relies upon the aforementioned
if-then-else constructs to effectively case-analyze the constraints instead of finding linear solutions.
For instance, if we use $\m{add}[m][n]$ only as $\m{add}[2][3]$ and $\m{add}[4][5]$, then instead of
outputting the expected expression $m+n$, Z3 outputs $\eif{m=2}{5}{9}$.

\paragraph*{\textbf{Binary Numbers} (bin)}
A more efficient representation of natural numbers is in their binary form.
Instead of sending $n$ $\mb{succ}$ messages, a type can send $\lceil \log_2 n \rceil$
$0$'s and $1$'s.
This representation is captured in this benchmark which introduces a type $\m{bin}[n]$
which can send labels $b0$, $b1$, or $e$, representing $0$, $1$, and termination, respectively.
We encode them in \emph{little endian} format, i.e. the least significant bit is sent first,
which makes implementations more convenient and types more intuitive.
The refinement indexes their value.
\begin{align*}
\mi{type}\ \m{bin}[n] = \ichoice{
	&\mb{b0}: \texists{n>0}{\exists k.\texists{n=2k}{\m{bin}[k]}},\\
	&\mb{b1}: \texists{n>0}{\exists k.\texists{n=2k+1}{\m{bin}[k]}},\\
	&\mb{e}: \texists{n=0}{\one}}
\end{align*}

\noindent The type of $\m{bin}[n]$ signifies that the provider can either send labels $\mb{b0}$, $\mb{b1}$, or $\mb{e}$.
In the case of $\mb{b0}$, the provider provides a proof that indeed $n > 0$ and produces a new number
$k$ such that $n = 2k$, meaning that $n$ is even.
Analogously for $\mb{b1}$, the provider asserts that $n > 0$ and is odd by producing $k$ such that
$n = 2k+1$.
Lastly, in the case of $\mb{e}$, the provider proves that $n = 0$ and terminates.
For this benchmark, we implement \textbf{successor} and \textbf{double}
with a few helper processes.

Due to the complexity of the $\m{bin}$ type,
this benchmark turns out to be the most challenging for our algorithm.
Inference only succeeds when all three optimizations are enabled,
thus demonstrating the efficacy of our optimizations in tandem with each other.
The types inferred are as expected and successfully typecheck:
\begin{lstlisting}[language=prest]
	decl successor[n] : (x : bin[n]) |- (y : bin[n+1])
	decl double[n] : (x : bin[n]) |- (y : bin[2*n])
\end{lstlisting}

\paragraph*{\textbf{Lists} (list)}

This is another practically important data structure, particularly
in the context of session types to store data.
Since we currently do not support inference of polymorphic session types,
we use Lisp-style lists of natural numbers, where the list is refined by its length.
The type $\m{list}[n]$ allows either sending the $\mb{cons}$ label if $n > 0$,
transitioning to $\m{list}[n-1]$, or the $\mb{nil}$ label if $n = 0$.
\begin{align*}
\mi{type}\ \m{list}[n] = \ichoice{
	&\mb{cons}: \texists{n>0}{\,\m{nat} \tensor \m{list}[n-1]},\\
	&\mb{nil}: \texists{n=0}{\one}}
\end{align*}

\noindent To focus on inference of lists, we use unrefined natural numbers ($\m{nat}$)
as the type of elements stored inside the list.
We implement 2 important (with some helper) list processes: \textbf{append}, which concatenates two lists,
and \textbf{split} which splits a list in half using two mutually recursive processes that
operate on even-length and odd-length lists.
We also include a $\m{listpair}[m][n]$ type, analogous to the aforementioned $\m{natpair}$ type,
to realize \textbf{split}.

We successfully find a solution only when
using transitivity and polynomial templates.
Compared to unary nats, Stage 1 takes more time while Stage 2 takes less.
We again observe a speedup from integer arithmetic as opposed to real arithmetic,
which lends itself to the hypothesis that integer arithmetic 
tends to find soltuions faster,
although real arithmetic is more useful for identifying solutions more consistently.
As expected, the types inferred are as follows:
\begin{lstlisting}[language=prest]
	decl append[m][n] : (x : list[m]) (y : list[n]) |- (z : list[m+n])
	decl split_even[n] : (x : list[2*n]) |- (y : listpair[n][n])
	decl split_odd[n] : (x : list[2*n+1]) |- (y : listpair[n+1][n])
\end{lstlisting}

\paragraph*{\textbf{Linear $\lambda$-calculus} (linlam)}
Our most complex benchmark is an implementation of a linear $\lambda$-calculus of expressions
on top of session types.
Due to its complexity, we present this example via 2 benchmarks: the first one does not contain
any refinements while the second one indexes types with their size.
Expressions are represented using the following type:
\[
	\mi{type} \; \m{exp} = \ichoice{\mb{lam}: \m{exp} \lolli \m{exp},\ \mb{app}: \m{exp} \tensor \m{exp}}
\]

\noindent Subtyping plays an important role here because the type of values written as
\[
	\mi{type} \; \m{val} = \ichoice{\mb{lam}: \m{exp} \lolli \m{exp}}
\]

\noindent is a \emph{subtype} of $\m{exp}$, i.e., the type of expressions: all values are also expressions.
We implement processes to evaluate and normalize expressions,
and we include a couple of simple programs ($\m{id}$, $\m{swap}$) in the calculus.

Immediately evident from the results is that transitivity here makes a significant difference:
only those trials with transitivity enabled succeed, while the others time out.
This is because of the presence of both $\m{val}$ and $\m{exp}$ types.
In previous examples, type assignments could immediately rule out a large swath of the
search space because there was only one type that would structurally fit a type variable,
but here $\m{exp}$ and $\m{val}$ are at times interchangeable.
Stage 1 involves much more work, which is exacerbated by the presence of intermediate type variables
if we forego transitivity.
Of course, since we lack refinements, Z3 has no constraints to solve, and so the
\textbf{strategy} and \textbf{arithmetic} optimizations, which only relate to Z3, have little observable effect.
Since Stage 2 concerns arithmetic refinements, of which there are none in this benchmark, the calls
to Z3 finish in $\sim$10ms.

Due to subtyping, there are multiple possible combinations of types that our algorithm could correctly infer. 
Our particular implementation infers the following key declarations, prioritizing the narrower type $\m{val}$ over its supertype $\m{exp}$:
\begin{lstlisting}[language=prest]
	decl apply : (e1 : exp) (e2 : exp) |- (e : exp)
	decl eval : (e : exp) |- (v : val)
	decl norm : (e : exp) |- (n : exp)
\end{lstlisting}

\paragraph*{\textbf{Sized Linear $\lambda$-calculus} (linlam-size)}
Finally, our most complex benchmark is extended with a refinement representing
the size of the term, as follows:
\begin{align*}
\m{type}\ \m{exp}[n] = \ichoice{
	&\mb{lam}: \texists{n>0}{\forall n'. \m{exp}[n'] \lolli \m{exp}[n + n' - 1]},\\
	&\mb{app}: \forall n_1 . \forall n_2 . \texists{n=n_1+n_2+1}{\m{exp}[n_1] \tensor \m{exp}[n_2]}
	}
\end{align*}

\noindent We adapt the aforementioned subtype $\m{val}$ to $\m{val}[n]$.
Finally, we introduce a new type, $\m{boundedVal}[n]=\exists k . \texists{k \leq n}{\m{val}[k]}$,
as the new return type of our $\m{eval}$ process:
although we cannot evaluate the exact size of an evaluated term, we can place an upper bound on it.

Like binary numbers,
this example reflects the full power of our optimizations.
Since this benchmark extends the unrefined linear $\lambda$-calculus,
it makes sense that disabling transitivity always leads to a timeout in Stage 1.
Since we also have arithmetic constraints to solve for this benchmark, our polynomial template
strategy proves crucial in preventing a Z3 timeout.
We successfully find solutions for both integer and real arithmetic with both other optimizations,
but integer arithmetic takes an order of magnitude longer to solve.

Our implementation returns the following key declarations:
\begin{lstlisting}[language=prest]
	decl apply[m][n] : (e1 : exp[m]) (e2 : exp[n]) |- (e : exp[m+n+1])
	decl eval : (e : exp[n]) |- (v : boundedVal[n])
\end{lstlisting}
Noteworthy here is that, via the inferred type of \textbf{eval}, our algorithm deduces
that the evaluation of an expression with size $n$ yields a value whose size is no greater than $n$.

\vspace{-1em}
\section{Related Works}\label{sec:related}
\vspace{-0.4em}

\paragraph*{\textbf{Techniques for Inferring Type Refinements}}
Numerous techniques have been proposed for inferring refinement types for
functional programs.
The earliest works on refinement type inference was in the context
of ML~\cite{pfenningRefinementTypes} where the types were refined
by a \emph{finite} programmer-specified lattice.
Finiteness was crucial for decidability of type inference, and type refinement
was performed iteratively until reaching a fixed point.
Our type system is considerably more expressive with a possibly infinite
set of arithmetic refinements, which makes inference undecidable.
Nonetheless, they introduced this technique of inference in stages:
first inferring base types followed by inferring type refinements.
This style was later adopted by Liquid Types~\cite{Rondon08PLDI}
where inference was reduced to Hindley-Milner type inference
for base types~\cite{Damas82POPL}, followed by liquid constraint generation
and constraint solving.
Type inference in this setting is decidable because they adopt a conservative
but decidable notion of subtyping, where subtyping of arbitrary
dependent types is reduced to a set of implication checks over base types.
In contrast, our notion of subtyping is general and therefore, undecidable.

Inference of type refinements has also been carried out using abstract
interpretation~\cite{Jhala11CAV} by reducing it to computing
invariants of simple, first-order imperative programs.
The \textsc{Fusion} algorithm~\cite{Cosman17ICFP} reduces inference
to finding a satisfying assignment for (nested) Horn Clause Constraints
in Negation Normal Form (NNF).
\citet{Hashimoto15SAS} reduce inference to a type optimization problem
which is further reduced to a Horn constraint optimization problem.
This technique can infer maximally preferred refinement types based on a user-specified
preference.
\citet{Pavlinovic21POPL} propose a more general type system
that is parametric both with the choice of the abstract refinement domain 
and context-sensitivity of control flow information.
More recently, learning-based techniques using randomized testing~\cite{Zhu15ICFP}
and LLMs~\cite{Sakkas25ICSE} have also been proposed for refinement inference.

The most distinguishing aspect of our style of refinements is that
session types follow \emph{structural} typing, whereas the aforementioned
systems follow \emph{nominal} typing.
With nominal typing where type equality and subtyping for base types depends
on the name, inferring the base type for all intermediate sub-expressions becomes
decidable and reduces to standard Hindley-Milner inference.
On the other hand, prior work~\cite{dasRefinements2020} shows that even though
linear arithmetic as well as type equality (and therefore subtyping) for basic
session types are decidable, the combination makes typing undecidable.
More concretely, there is no notion of a ``base type'' with session types and
therefore, our subtyping rules cannot be separated into subtyping for base
types and refinements; they mix together naturally.
As a result, even our first stage of inference differs from these works and eliminates
the intermediate types instead of inferring them.
More closely related to our work is the recent work on structural type
refinements~\cite{Binder22TYDE} which builds on algebraic subtyping to
combine properties of nominal and structural type systems.
The main difference between the two works is the design of the refinement layer.
While we refine types with arithmetic expressions, they use polymorphic variants,
thus supporting type application and union.
Due to this, we rely on a solver for inference while
they use constraint-based type inference.

\vspace{-0.5em}
\paragraph*{\textbf{Session Type Inference}}
Also related to our system are techniques for session type inference which
originated from observing communication primitives in
$\pi$-calculus~\cite{Graversen16WSFM,Pucella08Haskell}
implemented in OCaml~\cite{Imai22TACAS} and Haskell~\cite{Imai11PLACES}.
Session types have also been inferred with control flow information~\cite{Collingbourne10ENTCS}
and in a calculus of services and sessions~\cite{Mezzina08COORDINATION}.
\citet{Almeida25PLACES} propose an algorithm for inference of FreeST~\cite{Almeida19PLACES}
which is an implementation of context-free session types~\cite{Thiemann16ICFP}.
This technique builds on Quick Look~\cite{Serrano20ICFP} to enable inference of type annotations
in polymorphic applications.
Although FreeST allows non-regular protocols, arithmetic refinements are more general and can encode
stronger properties based on sizes and values that can even be non-linear.
Therefore, the inference algorithms are also quite different and FreeST does not require
solving arithmetic constraints.
To that end, none of these works support the DML-style~\cite{Xi99popl} of refinements like our system.

\vspace{-0.5em}
\paragraph*{\textbf{Session Subtyping}}
The concept of subtyping for session types has its roots in seminal work by \citet{gaySubtypingSessionTypes2005}, 
which introduces both the formal notion of a type simulation and a practical subtyping algorithm for basic session types in the $\pi$-calculus. 
Our work builds upon theirs by introducing refinements to the language. 
Crucially, their algorithm is both sound and complete, but a complete subtyping algorithm is impossible in the presence of refinements; thus, our algorithm is sound, but not complete. 
Session subtyping has been further explored in a variety of more specific contexts: for instance, \citet{Horne2024LAMP} develop and compare subtyping relation in both the isorecursive and equirecursive settings, and others \cite{Li24,Ghilezan23,Ghilezan19} propose subtyping for multiparty session types.
\citet{Mostrous15} propose subtyping for asynchronous session types in a higher-order $\pi$-calculus, extending the idea of a type simulation into an \emph{asynchronous} type simulation.
\citet{Bravetti2017IaC} show that asynchronous subtyping is undecidable in the presence of recursion.
In contrast to these works, ours is the first to explore session subtyping alongside a refinement system.
Our notion of refinements is adopted from \citet{dasVerified2020}, but their work was restricted to type equality and did not provide a type inference algorithm.

\vspace{-1em}
\section{Conclusion}
\vspace{-0.4em}

This paper presents a type inference algorithm for structural session types with arithmetic refinements. 
We develop a theoretical treatment of subtyping, introduce formal inference rules for constraint generation, 
and implement our algorithm in the Rast language. 
In addition, we detail the practical optimizations necessary to solve our constraints reliably,
and we demonstrate the benefits of these optimizations 
by evaluating the performance of our implementation on a number of examples.
In the future, we hope to extend our algorithm to other constructs in the Rast language,
namely \emph{temporal} and \emph{resource-aware} types \cite{Das20FSCD}.

\paragraph*{\textbf{Data Availability.}} 
The artifact associated
with this paper is publicly avialable 
at \url{https://doi.org/10.5281/zenodo.18155032}.

\bibliographystyle{splncs04nat}
\bibliography{stinf,references}

\newpage

\appendix
\section{Session Type Examples}\label{appendix-ex}
\paragraph*{\textbf{Closing and waiting.}} The process \lst{example1} has the following declaration:

\begin{lstlisting}[language=prest]
	decl example1 : (x : $\one$) (y : $\one$) |- (z : $\one$)
\end{lstlisting}

From its type, we deduce that it must wait for $x$ and $y$ to close, and then it closes the channel it offers, $z$. A suitable definition would be:

\begin{lstlisting}[language=prest]
	proc z <- example1 x y = wait x; wait y; close z
\end{lstlisting}

Alternatively, we could wait on $y$ before $x$.

\paragraph*{\textbf{Sending and receiving labels.}} We declare the process \lst{example1} as follows:

\begin{lstlisting}[language=prest]
	decl negate : (x: $\ichoice{\mb{true}: \one, \mb{false}: \one}$) |- (y: $\ichoice{\mb{true}: \one, \mb{false}: \one}$)
\end{lstlisting}

This process would receive a label from $x$, either $\m{true}$ or $\m{false}$, and then sends either $\m{true}$ or $\m{false}$ on $y$, depending on a case-analysis of what $x$ sent. We define \lst{negate} as:

\begin{lstlisting}[language=prest]
	proc y <- negate x = case x (
		true => y.false; wait y; close x
	  | false => y.true; wait y; close x
	)
\end{lstlisting}

The types do not tell us \emph{which} label \lst{negate} must send when— in the given definition, we send the opposite label, but we could also just always send $\mb{false}$; both programs would share the same declaration. However, the type system enforces that we cannot send on $x$ nor receive on $y$ a label outside of the label set $\{\mb{true},\mb{false}\}$.

\paragraph*{\textbf{Sending and receiving channels.}} An equivalent of function composition in our language could be written as the process $\m{comp}$ below:

\begin{lstlisting}[language=prest]
	decl comp : (x: $A \lolli B$), (y: $B \lolli C$) |- (z: $A \lolli C$)
	proc z <- comp x y =
		a <- recv z;
		send x a;
		send y x;
		z <-> y
\end{lstlisting}

We would first receive a channel $a$ of type $A$ from our provided channel $z$. We then send $a$ to $x$, which changes the type of $x$ into $B$. Now we send $x$ to $y$ to get $(y: C)$. We are left with $(y: C) \vdash (z: C)$, at which point we \emph{forward} between $z$ and $y$.

\paragraph*{\textbf{Type names and recursive types.}} Suppose our signature contains the type declaration $\m{nat} = \ichoice{\mb{zero}: \one, \mb{succ}: \m{nat}}$ Consider the following process:

\begin{lstlisting}[language=prest]
	decl consume : (x: $\m{nat}$) |- (y: $\one$)
	proc y <- consume x = case x (
		zero => y <-> x
	  | succ => y <- consume x
	)
\end{lstlisting}

Due to linearity, we can never disregard types in our context; if we wish to throw them away, we must instead transform them into a type which we know how to close. The process \lst{consume} does just that, taking a $\m{nat}$ and turning it into $\one$, which the client can then easily wait on. We implement this process by first case-analyzing the input: if $x$ sends $\mb{zero}$, we can simply forward, but if $x$ sends $\mb{succ}$, it recurses back to $\m{nat}$, so we \emph{spawn} $\m{consume}$ again with $x$ as input. By repeated respawns of $\m{consume}$, we "consume" all the labels $x$ sends until it turns into $\one$.

\paragraph*{\textbf{Assertion and assumption.}} Let our $\m{nat}$ type now refer to the declaration in a fresh $\Sigma$ that $\m{nat}[n] = \ichoice{\mb{zero}: \texists{n=0}{\one},\ \mb{succ}: \texists{n>0}{\m{nat}[n-1]}}$. Suppose we have a $\m{double}$ process which doubles the value of the input, declared as follows:

\begin{lstlisting}[language=prest]
	decl double[n] : (x: $\m{nat}$[n]) |- (y: $\m{nat}$[2*n])
	proc y <- double[n] x = case x (
		zero => assume x {n=0};
			y.zero; assert y {2*n=0};
			y <-> x
	  | succ => assume x {n>0};
	  		y.succ; assert y {2*n>0};
	  		y.succ; assert y {(2*n)-1>0};
	  		y <- double[n-1] x
	)
\end{lstlisting}

Note the refinement parameter $n$ in the process— when we spawn \lst{double}, we now specify an expression in place of $n$, which in this case corresponds to the value of the input channel $x$. When we receive a $\mb{zero}$ label from $x$, we also receive a proof that $n=0$; we use this fact to assert that $2*n=0$ on $y$ after sending $\mb{zero}$ on $y$. When we instead receive $\mb{succ}$ on $x$, we get that $n > 0$; our goal is to send $2$ $\mb{succ}$ messages on $y$. When we send the first, we use our assumption to assert that $2*n>0$, and the type of $y$ recurses to $\m{nat}[(2*n)-1]$ by our type definition. When we send a second $\mb{succ}$ label, we now must assert that $(2*n)-1 > 0$, which still follows from our assertion (and the fact that we work only in integers). We now have $(x: \m{nat}[n-1]) \vdash (y: \m{nat}[(2*n)-2]$, and we conclude by respawning $\m{double}$ with refinement $n-1$ instead of $n$.

\paragraph*{\textbf{Witnesses.}} Consider a type for binary numbers that sends bits $\m{b0}$ and $\m{b1}$, and terminates with $\m{e}$. If we want to express its value with a refinement, a natural thought would be to say:

\begin{align*}
\mi{type}\ \m{bin}[n] = \ichoice{
	&\mb{b0}: \texists{n>0}{\m{bin}[k / 2]},\\
	&\mb{b1}: \texists{n>0}{\m{bin}[(k-1) / 2]},\\
	&\mb{e}: \texists{n=0}{\one}}
\end{align*}

However, our arithmetic language lacks a division operator. Instead, we would use quantifiers to capture the same idea as follows:

\begin{align*}
\mi{type}\ \m{bin}[n] = \ichoice{
	&\mb{b0}: \texists{n>0}{\exists k.\texists{n=2k}{\m{bin}[k]}},\\
	&\mb{b1}: \texists{n>0}{\exists k.\texists{n=2k+1}{\m{bin}[k]}},\\
	&\mb{e}: \texists{n=0}{\one}}
\end{align*}

\vfill

\section{Validity}\label{appendix-valid}

We introduce three validity judgments, for signatures ($\vdash \Sigma\ \m{valid}$), declarations ($\vdash_\Sigma \Sigma'\ \m{valid}$), and types ($\mc{V} \semi \mc{C} \vdash_\Sigma A\ \m{valid}$).

\begin{mathpar}
\inferrule*[right=$\Sigma V$] {
  \vdash_\Sigma \Sigma\ \m{valid}
} {
  \vdash \Sigma\ \m{valid}
}
\and
\inferrule*[right=Nil$V$] { } {
  \vdash_\Sigma (\cdot)\ \m{valid}
}
\and
\inferrule*[right=ProcDecl$V$] {
    \vdash_\Sigma \Sigma'\ \m{valid}
} {
  \vdash_\Sigma \Sigma',\Delta \vdash f[\num{n}] :: (z: A)\ \m{valid}
}
\and
\inferrule*[right=ProcDef$V$] {
    \vdash_\Sigma \Sigma'\ \m{valid}
} {
  \vdash_\Sigma \Sigma',\espawn{x}{f[\num{n}]}{\num{y}} = P\ \m{valid}
}
\and
\inferrule*[right=TpDef$V$] {
  \vdash_\Sigma \Sigma'\ \m{valid}
  \and
  \num{n} \semi \top \vdash_\Sigma A\ \m{valid}
  \and
  A \neq V'[\num{e}]
} {
  \vdash_\Sigma \Sigma',V[\num{n}\ |\ \phi]=A\ \m{valid}
}
\and
\inferrule*[right=$\oplus V$] {
  (\forall \ell \in L)
  \and
  \mc{V} \semi \mc{C} \vdash_\Sigma A_\ell\ \m{valid}
} {
  \mc{V} \semi \mc{C} \vdash_\Sigma \ichoice{\ell: A_\ell}_{\ell \in L}\ \m{valid}
}
\and
\inferrule*[right=$\with V$] {
  (\forall \ell \in L)
  \and
  \mc{V} \semi \mc{C} \vdash_\Sigma A_\ell\ \m{valid}
} {
  \mc{V} \semi \mc{C} \vdash_\Sigma \echoice{\ell: A_\ell}_{\ell \in L}\ \m{valid}
}
\and
\inferrule*[right=$\tensor V$] {
  \mc{V} \semi \mc{C} \vdash_\Sigma A\ \m{valid}
  \and
  \mc{V} \semi \mc{C} \vdash_\Sigma B\ \m{valid}
} {
  \mc{V} \semi \mc{C} \vdash_\Sigma A \tensor B\ \m{valid}
}
\and
\inferrule*[right=$\lolli V$] {
  \mc{V} \semi \mc{C} \vdash_\Sigma A\ \m{valid}
  \and
  \mc{V} \semi \mc{C} \vdash_\Sigma B\ \m{valid}
} {
  \mc{V} \semi \mc{C} \vdash_\Sigma A \lolli B\ \m{valid}
}
\and
\inferrule*[right=$\one V$] { } {
  \mc{V} \semi \mc{C} \vdash_\Sigma \one\ \m{valid}
}
\and
\inferrule*[right=$?V$] {
  \mc{V} \semi \mc{C} \land \phi \vdash_\Sigma A\ \m{valid}
} {
  \mc{V} \semi \mc{C} \vdash_\Sigma \texists{\phi}{A}\ \m{valid}
}
\and
\inferrule*[right=$!V$] {
  \mc{V} \semi \mc{C} \land \phi \vdash_\Sigma A\ \m{valid}
} {
  \mc{V} \semi \mc{C} \vdash_\Sigma \tforall{\phi}{A}\ \m{valid}
}
\and
\inferrule*[right=$\exists V$] {
  \mc{V}, n \semi \mc{C} \vdash_\Sigma A\ \m{valid}
} {
  \mc{V} \semi \mc{C} \vdash_\Sigma \exists n.A\ \m{valid}
}
\and
\inferrule*[right=$\forall V$] {
  \mc{V}, n \semi \mc{C} \vdash_\Sigma A\ \m{valid}
} {
  \mc{V} \semi \mc{C} \vdash_\Sigma \forall n.A\ \m{valid}
}
\and
\inferrule*[right=def$V$] {
  V[\num{n}\ |\ \phi] \in \Sigma
  \and
  \mc{V} \semi \mc{C} \vDash \phi[\num{e} / \num{n}]
} {
  \mc{V} \semi \mc{C} \vdash_\Sigma V[\num{e}]\ \m{valid}
}
\end{mathpar}

\vfill

\section{Subtyping Rules}\label{appendix-st}
\begin{mathpar}
\inferrule*[right=$\m{st}_\oplus$] {
  (\forall \ell \in L) \and \mathcal{V} \semi \mathcal{C} \semi \Gamma \vdash A_\ell <: B_\ell
} {
  \mathcal{V} \semi \mathcal{C} \semi \Gamma \vdash \ichoice{\ell: A_\ell}_{\ell \in L} \st \ichoice{m: B_m}_{m \in M}
}
\and
\inferrule*[right=$\m{st}_\&$] {
  (\forall m \in M) \and \mathcal{V} \semi \mathcal{C} \semi \Gamma \vdash A_m \st B_m
} {
 \mathcal{V} \semi \mathcal{C} \semi \Gamma \vdash \echoice{\ell: A_\ell}_{\ell \in L} \st \echoice{m: B_m}_{m \in M}
}
\and
\inferrule*[right=$\m{st}_\tensor$] {
  \mathcal{V} \semi \mathcal{C} \semi \Gamma \vdash A_1 \st B_1
  \and 
  \mathcal{V} \semi \mathcal{C} \semi \Gamma \vdash A_2 \st B_2
} {
  \mathcal{V} \semi \mathcal{C} \semi \Gamma \vdash A_1 \tensor A_2 \st B_1 \tensor B_2
}
\and
\inferrule*[right=$\m{st}_\lolli$] {
  \mathcal{V} \semi \mathcal{C} \semi \Gamma \vdash B_1 \st A_1
  \and 
  \mathcal{V} \semi \mathcal{C} \semi \Gamma \vdash A_2 \st B_2
} {
  \mathcal{V} \semi \mathcal{C} \semi \Gamma \vdash A_1 \lolli A_2 \st B_1 \lolli B_2
}
\and
\inferrule*[right=$\m{st}_\one$] { } { 
  \mathcal{V} \semi \mathcal{C} \semi \Gamma \vdash \one \st \one
}
\and
\inferrule*[right=$\m{st}_?$] {
  \mathcal{V} \semi \mathcal{C} \vDash \phi \rightarrow \psi
  \and
  \mathcal{V} \semi \mathcal{C} \land \phi \semi \Gamma \vdash A \st B
} {
  \mathcal{V} \semi \mathcal{C} \semi \Gamma \vdash \texists{\phi}{A} \st \texists{\psi}{B}
}
\and
\inferrule*[right=$\m{st}_!$] {
  \mathcal{V} \semi \mathcal{C} \vDash \psi \rightarrow \phi
  \and
  \mathcal{V} \semi \mathcal{C} \land \psi \semi \Gamma \vdash A \st B
} {
  \mathcal{V} \semi \mathcal{C} \semi \Gamma \vdash \tforall{\phi}{A} \st \tforall{\psi}{B}
}
\and
\inferrule*[right=$\m{st}_{\exists}$] {
  \fresh{k}
  \and
  \mathcal{V}, k \semi \mathcal{C} \semi \Gamma \vdash A[k/m] \st B[k/n]
} {
  \mathcal{V} \semi \mathcal{C} \semi \Gamma \vdash \exists m . A \st \exists n . B
}
\and
\inferrule*[right=$\m{st}_{\forall}$] {
  \fresh{k}
  \and
  \mathcal{V}, k \semi \mathcal{C} \semi \Gamma \vdash A[k/m] \st B[k/n]
} {
  \mathcal{V} \semi \mathcal{C} \semi \Gamma \vdash \forall m . A \st \forall n . B
}
\and
\inferrule*[right=$\m{st}_\bot$] {
  \mathcal{V} \semi \mathcal{C} \vDash \bot
} {
  \mathcal{V} \semi \mathcal{C} \semi \Gamma \vdash A <: B
}
\and
\inferrule*[right=$\m{st}_\m{refl}$] {
  \mathcal{V} \semi \mathcal{C} \vDash e_1 = e_1' \land \cdots \land e_n = e_n'
} {
  \mathcal{V} \semi \mathcal{C} \semi \Gamma \vdash V[\num{e}] \st V[\num{e'}]
}
\and
\inferrule*[right=$\m{st}_\m{expd}$] {
  V_1[\num{v_1} \vert \phi_1] = A \in \Sigma
  \and 
  V_2[\num{v_2} \vert \phi_2] = B \in \Sigma
  \\\\
  \gamma = \langle \mathcal{V} \semi \mathcal{C} \semi V_1[\num{e_1}] \st V_2[\num{e_2}] \rangle
  \\\\
  \mathcal{V} \semi \mathcal{C} \semi \Gamma, \gamma \vdash A[\num{e_1}/\num{v_1}] \st B[\num{e_2}/\num{v_2}]
} {
  \mathcal{V} \semi \mathcal{C} \semi \Gamma \vdash V_1[\num{e_1}] \st V_2[\num{e_2}]
}
\and
\inferrule*[right=$\m{st}_\m{def}$] {
  \langle \mathcal{V}' \semi \mathcal{C}' \semi V_1[\num{e_1}'] \st V_2[\num{e_2}'] \rangle \in \Gamma
  \and
  \mathcal{V} \semi \mathcal{C} \vDash \exists \mathcal{V}' . \mathcal{C}' \land \num{e_1}' = \num{e_1} \land \num{e_2}' = \num{e_2}
} {
  \mathcal{V} \semi \mathcal{C} \semi \Gamma \vdash V_1[\num{e_1}] \st V_2[\num{e_2}]
}
\end{mathpar}

\vfill

\section{Soundness}\label{appendix-proof}

We first provide a series of auxiliary definitions to motivate a key lemma and our main proof.

\begin{definition}
For a substitution $\sigma$, we say $\mathcal{V} \semi \mathcal{C} \vDash \sigma$ to abbreviate that $\sigma$ is a ground substitution over $\mathcal{V}$ such that $\vDash \mathcal{C}[\sigma]$.
\end{definition}

\begin{definition}\label{forallvstc}
Given a relation $\mathcal{R}$ on valid ground types and two types $A, B$ with $\mathcal{V} \semi \mathcal{C} \vdash A, B\ \m{valid}$, we say $\forall \mathcal{V} . \mathcal{C} \rightarrow A \st_\mathcal{R} B$ if, for all substitutions $\sigma$ with $\mathcal{V} \semi \mathcal{C} \vDash \sigma$, we get $(A[\sigma], B[\sigma]) \in \mathcal{R}$.

We say $\forall \mathcal{V} . \mathcal{C} \rightarrow A \st B$ if there exists a type simulation $\mathcal{R}$ satisfying $\forall \mathcal{V} . \mathcal{C} \rightarrow A \st_\mathcal{R} B$.
\end{definition}

\begin{lemma}\label{thelemma}
Suppose $\forall \mathcal{V}' . \mathcal{C}' \rightarrow V_1[\num{e_1}'] \st_\mathcal{R} V_2[\num{e_2}']$, and assume that $\mathcal{V} \semi  \mathcal{C} \vDash \exists \mathcal{V}' . \mathcal{C}' \land \num{e_1}'=\num{e_1} \land \num{e_2}'=\num{e_2}$. Then it follows that $\forall \mathcal{V} . \mathcal{C} \rightarrow V_1[\num{e_1}] \st_\mathcal{R} V_2[\num{e_2}]$.
\end{lemma}

\begin{proof}
By definition of $\st_\mathcal{R}$, it suffices to show that, for all substitutions $\sigma$ such that $\mathcal{V} \semi \mathcal{C} \vDash \sigma$, we get $V_1[\num{e_1}[\sigma]] \st_\mathcal{R} V_2[\num{e_2}[\sigma]]$.

Take an arbitrary such $\sigma$: since $\vDash \mathcal{C}[\sigma]$, we can apply $\sigma$ to our second assumption, yielding that $\exists \mathcal{V}' . \mathcal{C}' \land \num{e_1}'=\num{e_1}[\sigma] \land \num{e_2}'=\num{e_2}[\sigma]$. By definition, there is \emph{some} ground substitution $\sigma'$ over $\mathcal{V}'$ such that $\vDash \mathcal{C}'[\sigma']$, $\num{e_1}'[\sigma']=\num{e_1}[\sigma]$, and $\num{e_2}'[\sigma']=\num{e_2}[\sigma]$.

From our first assumption, for \emph{any} such $\sigma'$ with $\vDash C'[\sigma']$, we get $V_1[\num{e_1}'[\sigma']] \st_\mathcal{R} V_2[\num{e_2}'[\sigma']]$. Since $\num{e_1}'[\sigma']=\num{e_1}[\sigma]$ and $\num{e_2}'[\sigma']=\num{e_2}[\sigma]$, we get $V_1[\num{e_1}[\sigma]] \st_\mathcal{R} V_2[\num{e_2}[\sigma]]$.

Since this applies for any such $\sigma$, we are done as described above: by definition, $\forall \mathcal{V} . \mathcal{C} \rightarrow V_1[\num{e_1}] \st_\mathcal{R} V_2[\num{e_2}]$.
\end{proof}

Recall the statement of Theorem \ref{thm:sound}: if $\mathcal{V}_0 \semi \mathcal{C}_0 \semi \cdot \Vdash A_0 \st B_0$, then $\forall \mathcal{V}_0 . \mc{C}_0 \rightarrow A_0 \st B_0$.

\paragraph*{\textbf{Proof of Theorem \ref{thm:sound}.}} From the antecedent we get a derivation $\mathcal{D}_0$ of $\mathcal{V}_0, \mathcal{C}_0, \cdot \Vdash A_0 \st B_0$. Define $\mathcal{R}$ on closed valid types as:
$$\mathcal{R} = \{(A[\sigma], B[\sigma])\ \vert\ \langle \mathcal{V} \semi \mathcal{C} \semi A \st B \rangle \in S(\mathcal{D}_0) \text{ and } \mathcal{V} \semi \mathcal{C} \vDash \sigma\}$$

We will show $\mathcal{R}$ is a type simulation. To do so, we consider arbitrary $(A[\sigma], B[\sigma]) \in \mathcal{R}$; by definition of $\mathcal{R}$, there must be some closure $\langle \mathcal{V} \semi \mathcal{C} \semi A \st B\rangle \in S(\mathcal{D}_0)$ and some $\sigma$ with $\mathcal{V} \semi \mathcal{C} \vDash \sigma$.

Consider first the case where $\mathcal{V} \semi \mathcal{C} \vDash \bot$. It follows from the $\m{st}_\bot$ rule that $\langle \mathcal{V} \semi \mathcal{C} \semi A \st B \rangle \in S(\mathcal{D}_0)$. Furthermore, $\forall \mathcal{V} . \mathcal{C} \rightarrow A \st B$ is vacuously true, and so soundness holds.

If instead $\mathcal{V} \semi \mathcal{C} \not \vDash \bot$, it follows that there exists some ground substitution $\sigma$ on $\mathcal{V}$ that satisfies $\mathcal{C}$, i.e. $\mathcal{V} \semi \mathcal{C} \vDash \sigma$. We proceed by case-analysis on the structure of $A$ with an arbitrary such $\sigma$. Most cases follow by simple structural analysis; we include a subset for demonstrative purposes here.

If $A=\ichoice{\ell : A_\ell}_{\ell \in L}$, then by enumeration of rules, we must have $B=\ichoice{m: B_m}_{m \in M}$. It follows from $\m{st}_\oplus$ that, for all $\ell \in L$, we get $\langle \mathcal{V} \semi \mathcal{C} \semi A_\ell \st B_\ell \rangle \in S(\mathcal{D}_0)$. By definition of $\mathcal{R}$, we then have $(A_\ell[\sigma], B_\ell[\sigma]) \in \mathcal{R}$. Since $A[\sigma]=\ichoice{\ell : A_\ell[\sigma]}_{\ell \in L}$ and $B[\sigma]=\ichoice{m : B_m[\sigma]}_{m \in M}$, we conclude that, for an arbitrary $\sigma$, if $(A[\sigma], B[\sigma]) \in \mathcal{R}$, then $(A_\ell[\sigma], B_\ell[\sigma]) \in \mathcal{R}$ for all $\ell \in L$. Thus, we have satisfied the condition for a type simulation.

If $A=\texists{\phi}A'$, then by enumeration of rules, we must have $B=\texists{\psi}B'$. It follows from $\m{st}_?$ that we get $\langle \mathcal{V} \semi \mathcal{C} \land \phi \semi A' \st B' \rangle \in S(\mathcal{D}_0)$, as well as the semantic judgment $\mathcal{V} \semi \mathcal{C} \vDash \phi \rightarrow \psi$. Since our considered $\sigma$ satisfies $\vDash \mathcal{C}[\sigma]$, we have two sub-cases: either $\vDash \phi[\sigma]$ or $\not \vDash \phi[\sigma]$. If $\vDash \phi[\sigma]$, it must follow that $\vDash (\mathcal{C} \land \phi)[\sigma]$, so by definition of $\mathcal{R}$ we get $(A'[\sigma], B'[\sigma]) \in \mathcal{R}$. Furthermore, since $\vDash \phi[\sigma]$ and $\mathcal{V} \semi \mathcal{C} \vDash \phi \rightarrow \psi$, we also have $\vDash \psi[\sigma]$. Since $A[\sigma]=\texists{\phi[\sigma]}A'$ and $B[\sigma]=\texists{\psi[\sigma]}B'$, we satisfy condition (1) for a type simulation. If instead $\not \vDash \phi[\sigma]$, then we trivially satisfy condition (2) for a type simulation; thus, in either subcase, it holds that $\mathcal{R}$ is a type simulation.

If $A=\exists m.A'$, then by enumeration of rules, $B=\exists n.B'$. It follows from $\m{st}_\exists$ that we get $\langle \mathcal{V}, k \semi \mathcal{C} \semi A'[k/m] \st B'[k/m] \rangle \in S(\mathcal{D}_0)$ for some fresh $k$. Since $k$ is fresh, $k \notin C$, and so for any $i \in \mathbb{N}$ we get $(A'[\sigma, i/k], B'[\sigma, i/k]) \in \mathcal{R}$, thereby satisfying that $\mathcal{R}$ is a type simulation.

If $A=V_1[\num{e}_1]$, then we have three sub-cases: either the $\m{st}_\m{refl}$ rule, the $\m{st}_\m{def}$ rule, or the $\m{st}_\m{expd}$ rule applies; we consider them in that order. If we apply the $\m{st}_\m{def}$ rule, then we do not add anything to $S(\mathcal{D}_0)$ by definition of $\mc{R}$, but in fact doing so is unnecessary. The first premise of the $\m{st}_\m{def}$ rule tells us that we have already seen $\langle \mathcal{V}' \semi \mathcal{C}' \semi V_1[\num{e_1}'] \st V_2[\num{e_2}']\rangle$, implying that $\forall \mathcal{V}' . \mathcal{C}' \rightarrow  V_1[\num{e_1}'] \st_\mathcal{R} V_2[\num{e_2}']$. It follows from Lemma \ref{thelemma} and the second premise to the $\m{st}_\m{def}$ rule that $V_1[\num{e_1}[\sigma]] \st_\mathcal{R} V_2[\num{e_2}[\sigma]]$, so therefore $(V_1[\num{e_1}[\sigma]], V_2[\num{e_2}[\sigma]]) \in \mathcal{R}$.

We now conclude our proof. Since our derivation $\mathcal{D}_0$ must prove $\mathcal{V}_0, \mathcal{C}_0, \cdot \Vdash A \st B$, we have $\langle \mathcal{V}_0 \semi \mathcal{C}_0 \semi A_0 \st B_0\rangle \in S(\mathcal{D}_0)$ by definition of $S$. Then, by definition of $\mathcal{R}$, we have $(A_0[\sigma], B_0[\sigma]) \in \mathcal{R}$ for any $\sigma$ over $\mathcal{V}_0$ satisfying $\vDash \mathcal{C}_0[\sigma]$. By Definition \ref{forallvstc}, we then say $\forall \mathcal{V}_0 . \mathcal{C}_0 \rightarrow A_0 \st_\mathcal{R} B_0$, and finally, since $\mathcal{R}$ is a type simulation, we say $\forall \mathcal{V}_0 . \mathcal{C}_0 \rightarrow A_0 \st B_0$ and we are done.

\vfill

\section{Constraint Generation Rules}\label{appendix-cgen}

\begin{mathpar}
\inferrule*[right=$\oplus$R] {
  \mc{V} \semi \mc{C} \Vdash \ichoice{k: \textcolor{blue}{B}} \st A
  \and
  \mathcal{V} \semi \mathcal{C} \semi \Delta \vdash P :: (x: \textcolor{blue}{B})
} {
  \mathcal{V} \semi \mathcal{C} \semi \Delta \vdash x.k \semi P :: (x: A)
}
\and
\inferrule*[right=$\oplus$L] {
  \mc{V} \semi \mc{C} \Vdash A \st \ichoice{\ell: \textcolor{blue}{A_\ell}}_{\ell \in L}
  \and 
  (\forall \ell \in L) \quad
  \mathcal{V} \semi \mathcal{C} \semi \Delta, (x: \textcolor{blue}{A_\ell}) \vdash Q_\ell :: (z: C)
} {
  \mathcal{V} \semi \mathcal{C} \semi \Delta, (x: A) \vdash \ecase{x}{\ell}{Q_\ell} :: (z: C)
}
\and
\inferrule*[right=$\with$R] {
  \mc{V} \semi \mc{C} \Vdash \echoice{\ell: \textcolor{blue}{A_\ell}}_{\ell \in L} \st A
  \and
  (\forall \ell \in L) \quad
  \mathcal{V} \semi \mathcal{C} \semi \Delta \vdash P_\ell :: , (x: \textcolor{blue}{A_\ell})
} {
  \mathcal{V} \semi \mathcal{C} \semi \Delta \vdash \ecase{x}{\ell}{P_\ell} :: (x: A)
}
\and
\inferrule*[right=$\with$L] {
  \mc{V} \semi \mc{C} \Vdash A \st \echoice{k: \textcolor{blue}{A_k}}
  \and
  \mathcal{V} \semi \mathcal{C} \semi \Delta, (x: \textcolor{blue}{A_k}) \vdash Q :: (z: C)
} {
  \mathcal{V} \semi \mathcal{C} \semi \Delta, (x: A) \vdash x.k \semi P :: (z: C)
}
\and
\inferrule*[right=$\tensor$R] {
  \mc{V} \semi \mc{C} \Vdash A_1 \tensor \textcolor{blue}{A_2} \st A
  \and
  \mathcal{V} \semi \mathcal{C} \semi \Delta \vdash P :: (x: \textcolor{blue}{A_2})
} {
  \mathcal{V} \semi \mathcal{C} \semi \Delta, (y: A_1) \vdash \esend{x}{y} \semi P :: (x : A)
}
\and
\inferrule*[right=$\tensor$L] {
  \mc{V} \semi \mc{C} \Vdash A \st \textcolor{blue}{A_1} \tensor \textcolor{blue}{A_2}
  \and
  \mathcal{V} \semi \mathcal{C} \semi \Delta, (y: \textcolor{blue}{A_1}), (x: \textcolor{blue}{A_2}) \vdash Q :: (z: C)
} {
  \mathcal{V} \semi \mathcal{C} \semi \Delta, (x: A) \vdash \erecv{x}{y} \semi Q :: (z : C)
}
\and
\inferrule*[right=$\lolli$R] {
  \mc{V} \semi \mc{C} \Vdash \textcolor{blue}{A_1} \lolli \textcolor{blue}{A_2}  \st A
  \and
  \mathcal{V} \semi \mathcal{C} \semi \Delta, (y: \textcolor{blue}{A_1}) \vdash P :: (x: \textcolor{blue}{A_2})
} {
  \mathcal{V} \semi \mathcal{C} \semi \Delta \vdash \erecv{x}{y} \semi P :: (x : A)
}
\and
\inferrule*[right=$\lolli$L] {
  \mc{V} \semi \mc{C} \Vdash A \st A_1 \lolli \textcolor{blue}{A_2}
  \and
  \mathcal{V} \semi \mathcal{C} \semi \Delta, (x: \textcolor{blue}{A_2}) \vdash Q :: (z: C)
} {
  \mathcal{V} \semi \mathcal{C} \semi \Delta, (x: A), (y: A_1) \vdash \esend{x}{y} \semi Q :: (z : C)
}
\and
\inferrule*[right=$\one$R] {
  \mathcal{V} \semi \mathcal{C} \Vdash \one \st A
} {
  \mathcal{V} \semi \mathcal{C} \semi \cdot \vdash \eclose{x} :: (x: A)
}
\and
\inferrule*[right=$\one$L] {
  \mathcal{V} \semi \mathcal{C} \Vdash A \st \one
  \and
  \mathcal{V} \semi \mathcal{C} \semi \Delta \vdash Q :: (z: C)
} {
  \mathcal{V} \semi \mathcal{C} \semi \Delta, (x: \one) \vdash \ewait{x} \semi Q :: (z: C)
}
\and
\inferrule*[right=$\m{id}$] {
  \mathcal{V} \semi \mathcal{C} \Vdash A \st B
} {
  \mathcal{V} \semi \mathcal{C} \semi (y: A) \vdash \fwd{x}{y} :: (x: B)
}
\and
\inferrule*[right=$\m{def}$] {
  (\num{y_i' : B_i'})_{i \in I} \vdash f[\num{n}\: |\: \phi] = P_f :: (x' : A') \in \Sigma
  \and\\
  \mathcal{V} \semi \mathcal{C} \land \phi[\num{e}/\num{n}] \Vdash A'[\num{e}/\num{n}] \st  \textcolor{blue}{A}
  \and
  (i \in I) \quad \mathcal{V} \semi \mathcal{C} \land \phi[\num{e}/\num{n}] \Vdash B_i \st B_i' [\num{e}/\num{n}]
  \and\\
  \mathcal{V} \semi \mathcal{C} \semi \Delta, (x:  \textcolor{blue}{A}[\num{e}/\num{n}]) \vdash Q :: (z: C)
} {
  \mathcal{V} \semi \mathcal{C} \semi \Delta, (\num{y_i : B_i})_{i \in I} \vdash \espawn{x}{f[\num{e}]}{\num{y}} \semi Q :: (z: C)
}
\end{mathpar}
\vfill
\begin{mathpar}
\inferrule*[right=?R] {
  \mathcal{V} \semi \mathcal{C} \vDash \phi
  \and
  \mathcal{V} \semi \mathcal{C} \Vdash \texists{\phi}{\textcolor{blue}{A'}} \st A
  \and
  \mathcal{V} \semi \mathcal{C} \semi \Delta \vdash P :: (x: \textcolor{blue}{A'})
} {
  \mathcal{V} \semi \mathcal{C} \semi \Delta \vdash \eassert{x}{\phi} \semi P :: (x: A)
}
\and
\inferrule*[right=?L] {
  \mathcal{V} \semi \mathcal{C} \Vdash A \st \texists{\phi}{\textcolor{blue}{A'}}
  \and
  \mathcal{V} \semi \mathcal{C} \land \phi \semi \Delta, (x: \textcolor{blue}{A'}) \vdash Q :: (z: C)
} {
  \mathcal{V} \semi \mathcal{C} \semi \Delta, (x: A) \vdash \eassume{x}{\phi} \semi Q :: (z: C)
}
\and
\inferrule*[right=!R] {
  \mathcal{V} \semi \mathcal{C} \Vdash \tforall{\phi}{\textcolor{blue}{A'}} \st A
  \and
  \mathcal{V} \semi \mathcal{C} \land \phi \semi \Delta \vdash P :: (x: \textcolor{blue}{A'})
} {
  \mathcal{V} \semi \mathcal{C} \semi \Delta \vdash \eassume{x}{\phi} \semi P :: (x: A)
}
\and
\inferrule*[right=!L] {
  \mathcal{V} \semi \mathcal{C} \vDash \phi
  \and
  \mathcal{V} \semi \mathcal{C} \Vdash A \st \tforall{\phi}{\textcolor{blue}{A'}}
  \and
  \mathcal{V} \semi \mathcal{C} \semi \Delta, (x: \textcolor{blue}{A'}) \vdash Q :: (z: C)
} {
  \mathcal{V} \semi \mathcal{C} \semi \Delta, (x: A) \vdash \eassert{x}{\phi} \semi Q :: (z: C)
}
\and
\inferrule*[right=$\exists$R] {
  \mathcal{V} \semi \mathcal{C} \vDash e \geq 0
  \and
  \mathcal{V} \semi \mathcal{C} \Vdash \exists n.\textcolor{blue}{A'} \st A
  \and
  \mathcal{V} \semi \mathcal{C} \land n=e \semi \Delta \vdash P :: (x: \textcolor{blue}{A'})
} {
  \mathcal{V} \semi \mathcal{C} \semi \Delta \vdash \esend{x}{\{e\}} \semi P :: (x: A)
}
\and
\inferrule*[right=$\exists$L] {
  \mathcal{V} \semi \mathcal{C} \Vdash A \st \exists n.\textcolor{blue}{A'}
  \and
  \mathcal{V}, n \semi \mathcal{C} \semi \Delta, (x: \textcolor{blue}{A'}) \vdash Q_n :: (z: C)
} {
  \mathcal{V} \semi \mathcal{C} \semi \Delta, (x: A) \vdash \erecv{x}{\{n\}} \semi Q_n :: (z: C)
}
\and
\inferrule*[right=$\forall$R] {
  \mathcal{V} \semi \mathcal{C} \semi \forall n.\textcolor{blue}{A'} \st A
  \and
  \mathcal{V}, n \semi \mathcal{C} \semi \Delta \vdash P_n :: (x: \textcolor{blue}{A'})
} {
  \mathcal{V} \semi \mathcal{C} \semi \Delta \vdash \erecv{x}{\{n\}} \semi P_n :: (x: A)
}
\and
\inferrule*[right=$\forall$L] {
  \mathcal{V} \semi \mathcal{C} \vDash e \geq 0
  \and
  \mathcal{V} \semi \mathcal{C} \Vdash A \st \forall n.\textcolor{blue}{A'}
  \and
  \mathcal{V} \semi \mathcal{C} \land n=e \semi \Delta, (x: \textcolor{blue}{A'}) \vdash Q :: (z: C)
} {
  \mathcal{V} \semi \mathcal{C} \semi \Delta, (x: A) \vdash \esend{x}{\{e\}} \semi Q :: (z: C)
}
\end{mathpar}

\vfill

\section{Typechecking Rules}\label{appendix-tc}

\begin{mathpar}
\inferrule*[right=tc-$\oplus$R] {
  (k \in L)
  \and
  \mathcal{V} \semi \mathcal{C} \semi \Delta \vdash P :: (x: A_k)
} {
  \mathcal{V} \semi \mathcal{C} \semi \Delta \vdash x.k \semi P :: (x: \ichoice{\ell: A_\ell}_{\ell \in L})
}
\and
\inferrule*[right=tc-$\oplus$L] {
  (\forall \ell \in L)
  \and
  \mathcal{V} \semi \mathcal{C} \semi \Delta, (x: A_\ell) \vdash P :: (z: C)
} {
  \mathcal{V} \semi \mathcal{C} \semi \Delta, (x:  \ichoice{\ell: A_\ell}_{\ell \in L}) \vdash \ecase{x}{\ell}{Q_\ell} :: (z: C)
}
\and
\inferrule*[right=tc-$\with$R] {
  (\forall \ell \in L)
  \and
  \mathcal{V} \semi \mathcal{C} \semi \Delta \vdash P :: (x: A_\ell)
} {
  \mathcal{V} \semi \mathcal{C} \semi \Delta \vdash \ecase{x}{\ell}{Q_\ell} :: (x: \echoice{\ell: A_\ell}_{\ell \in L})
}
\and
\inferrule*[right=tc-$\with$L] {
  (k \in L)
  \and
  \mathcal{V} \semi \mathcal{C} \semi \Delta, (x: A_k) \vdash Q :: (z: C)
} {
  \mathcal{V} \semi \mathcal{C} \semi \Delta, (x: \echoice{\ell: A_\ell}_{\ell \in L}) \vdash x.k \semi Q :: (z: C)
}
\and
\inferrule*[right=tc-$\tensor$R] {
  \mathcal{V} \semi \mathcal{C} \semi \Delta \vdash P :: (x: B)
} {
  \mathcal{V} \semi \mathcal{C} \semi \Delta, (e: A) \vdash \esend{x}{e} \semi P :: (x: A \tensor B)
}
\and
\inferrule*[right=tc-$\tensor$L] {
  \mathcal{V} \semi \mathcal{C} \semi \Delta, (x: B), (y: A) \vdash Q :: (z: C)
} {
  \mathcal{V} \semi \mathcal{C} \semi \Delta, (x: A \tensor B) \vdash \erecv{x}{y} \semi Q :: (z: C)
}
\and
\inferrule*[right=tc-$\lolli$R] {
  \mathcal{V} \semi \mathcal{C} \semi \Delta, (y: A) \vdash P :: (x: B)
} {
  \mathcal{V} \semi \mathcal{C} \semi \Delta \vdash \erecv{x}{y} \semi P :: (x: A \lolli B)
}
\and
\inferrule*[right=tc-$\lolli$L] {
  \mathcal{V} \semi \mathcal{C} \semi \Delta, (x: B) \vdash Q :: (z: C)
} {
  \mathcal{V} \semi \mathcal{C} \semi \Delta, (x: A \lolli B), (e: A) \vdash \esend{x}{e} \semi Q :: (z: C)
}
\and
\inferrule*[right=tc-$\one$R] { } {
  \mathcal{V} \semi \mathcal{C} \semi \cdot \vdash \eclose{x} :: (x: \one)
}
\and
\inferrule*[right=tc-$\one$L] {
    \mc{V} \semi \mc{C} \semi \Delta \vdash Q :: (z: C)
} {
  \mathcal{V} \semi \mathcal{C} \semi \Delta, (x: \one) \vdash \ewait{x} \semi Q :: (z: C)
}
\and
\inferrule*[right=tc-id] { } {
  \mathcal{V} \semi \mathcal{C} \semi (y: A) \vdash \fwd{x}{y} :: (x: A)
}
\and
\inferrule*[right=tc-def] {
  (\num{y_i' : B_i})_{i \in I} \vdash f[\num{n}] = P_f :: (x' : A) \in \Sigma
  \and
  \mc{V} \semi \mc{C} \semi \Delta, (x: A[\num{e}/\num{n}]) \vdash Q :: (z: C)
  } {
  \mathcal{V} \semi \mathcal{C} \semi \Delta, \num{(y_i: B_i)}_{i \in I} \vdash \espawn{x}{f[\num{e}]}{\num{y}} \semi Q :: (z: C)
}
\end{mathpar}
\vfill
\begin{mathpar}
\inferrule*[right=tc-?R] {
  \mc{V} \semi \mc{C} \vDash \phi
  \and
  \mathcal{V} \semi \mathcal{C} \semi \Delta \vdash P :: (x: A)
} {
  \mathcal{V} \semi \mathcal{C} \semi \Delta \vdash \eassert{x}{\phi} \semi P :: (x: \texists{\phi}{A})
}
\and
\inferrule*[right=tc-?L] {
  \mathcal{V} \semi \mathcal{C} \land \phi \semi \Delta, (x: A) \vdash Q :: (z: C)
} {
  \mathcal{V} \semi \mathcal{C} \semi \Delta, (x: \texists{\phi}{A}) \vdash \eassume{x}{\phi} \semi Q :: (z: C)
}
  \and
\inferrule*[right=tc-!R] {
  \mathcal{V} \semi \mathcal{C} \land \phi \semi \Delta \vdash P :: (x: A)
} {
  \mathcal{V} \semi \mathcal{C} \semi \Delta \vdash \eassume{x}{\phi} \semi P :: (x: \tforall{\phi}{A})
}
\and
\inferrule*[right=tc-!L] {
  \mc{V} \semi \mc{C} \vDash \phi
  \and
  \mathcal{V} \semi \mathcal{C} \semi \Delta, (x: A) \vdash Q :: (z: C)
} {
  \mathcal{V} \semi \mathcal{C} \semi \Delta, (x: \tforall{\phi}{A}) \vdash \eassert{x}{\phi} \semi Q :: (z: C)
}
\and
\inferrule*[right=tc-$\exists$R] {
  \mc{V} \semi \mc{C} \vdash e : \m{nat}
  \and
  \mathcal{V} \semi \mathcal{C} \semi \Delta \vdash P :: (x: A[e/n])
} {
  \mathcal{V} \semi \mathcal{C} \semi \Delta \vdash \esend{x}{\{e\}} \semi P :: (x: \exists n . A)
}
\and
\inferrule*[right=tc-$\exists$L] {
  \mathcal{V}, n \semi \mathcal{C} \semi \Delta, (x: A) \vdash Q_n :: (z: C)
} {
  \mathcal{V} \semi \mathcal{C} \semi \Delta, (x: \forall n . A) \vdash \erecv{x}{\{n\}} \semi Q_n :: (z: C)
}
\and
\inferrule*[right=tc-$\forall$R] {
  \mathcal{V}, n \semi \mathcal{C} \semi \Delta \vdash P_n :: (x: A)
} {
  \mathcal{V} \semi \mathcal{C} \semi \Delta \vdash \erecv{x}{\{n\}} \semi P_n :: (x: \forall n . A)
}
\and
\inferrule*[right=tc-$\forall$L] {
  \mc{V} \semi \mc{C} \vdash e : \m{nat}
  \and
  \mathcal{V} \semi \mathcal{C} \semi \Delta, (x: A[e/n]) \vdash Q :: (z: C)
} {
  \mathcal{V} \semi \mathcal{C} \semi \Delta, (x: \exists n.A) \vdash \esend{x}{\{e\}} \semi Q :: (z: C)
}
\end{mathpar}

\end{document}